%% file: quantised_estimator_TSP.tex
\documentclass[journal]{IEEEtran}
\ifCLASSINFOpdf
\else
\fi

\usepackage{cite}
\usepackage{amsmath,amssymb,amsfonts, amsthm}
\usepackage{algorithmic}
\usepackage{graphicx}
\usepackage{textcomp}
\usepackage{xcolor}
\usepackage{bbm}
\usepackage{subcaption}
\usepackage{mwe}
\usepackage{url}
\usepackage{multirow}
\usepackage{float}
\usepackage{enumitem}
\usepackage{pifont}
\usepackage{xfrac}

\usepackage{dcolumn}
\newcolumntype{d}[1]{D{.}{.}{#1}}

\newtheorem{theorem}{Theorem}[section]

\newtheorem{definition}{Definition}[section]
\newtheorem*{remark}{Remark}

\allowdisplaybreaks

\newcommand{\raisedchi}{\raisebox{\depth}{\(\chi\)}}
\newcommand{\ie}{{\em i.e.,} }
\newcommand{\eg}{{\em e.g.,} }

\begin{document}
%
\title{NPD Entropy: A Non-Parametric Differential Entropy Rate Estimator}
%
%
%

\author{Andrew~Feutrill 
        and Matthew~Roughan,~\IEEEmembership{Fellow,~IEEE,}
\thanks{A. Feutrill is with CSIRO's Data 61, the School of Mathematical Sciences, University of Adelaide, Adelaide, Australia and the Australian Research Council's Centre of Excellence for Mathematical and Statistical Frontiers (ACEMS), e-mail: andrew.feutrill@data61.csiro.au.}
\thanks{M. Roughan is with the School of Mathematical Sciences, University of Adelaide, Adelaide, Australia and the Australian Research Council's Centre of Excellence for Mathematical and Statistical Frontiers (ACEMS).}}

%
%

\markboth{}
{Feutrill \MakeLowercase{\textit{et al.}}: }
%



\maketitle

\begin{abstract}
The estimation of entropy rates for stationary
discrete-valued stochastic processes is a well studied problem in
information theory. However, estimating the entropy rate for
stationary continuous-valued stochastic processes has not received as
much attention. In fact, many current techniques are not able to
accurately estimate or characterise the complexity of the differential
entropy rate for strongly correlated
processes, such as Fractional Gaussian Noise and ARFIMA(0,d,0). To the
point that some cannot even detect the trend of the entropy rate, \eg when it increases/decreases, maximum, or asymptotic trends, as a function of their Hurst parameter. However, a recently developed technique provides accurate estimates at a high computational cost. In this paper, we define a robust technique for non-parametrically estimating the differential entropy rate of a continuous valued stochastic process from observed data, by making an explicit link between the differential entropy rate and the Shannon entropy rate of a quantised version of the original data. Estimation is performed by a Shannon entropy rate estimator, and then converted to a differential entropy rate estimate. We show that this technique inherits many important statistical properties from the Shannon entropy rate estimator. The estimator is able to provide better estimates than the defined relative measures and much quicker estimates than known absolute measures, for strongly correlated processes. Finally, we analyse the complexity of the estimation technique and test the robustness to non-stationarity, and show that none of the current techniques are robust to non-stationarity, even if they are robust to strong correlations.
\end{abstract}


%
\IEEEpeerreviewmaketitle

\input{intro_TSP}

\input{background_TSP}

\input{discrete_time_estimators_TSP}

\input{continuous_time_estimators_TSP}

\input{shannon_differential_link_TSP}

\input{quantised_estimation_TSP}

\input{conclusion}


%



\section*{Acknowledgment}

This research was funded by CSIRO's Data61, the ARC Centre of Excellence in Mathematical and Statistical Frontiers and Defence Science and Technology Group.

\ifCLASSOPTIONcaptionsoff
  \newpage
\fi



%

\bibliographystyle{abbrv}
\bibliography{myBibliography}

%








\end{document}

%% file: intro_TSP.tex
\section{Introduction}
%
%
%
%
\IEEEPARstart{E}{stimation} of entropy rate is a classical problem in
information theory.  The entropy rate is the asymptotic limit of the
average per sample entropy of a discrete-time stationary stochastic
process. It is used as a measure of the complexity of the process and
thus to perform comparisons.

Estimation of entropy rate is easy when the underlying stochastic
process is known to follow a simple model such as a Markov chain. The
reality is that most real data sequences are not trivial to model.  Non-parametric approaches to estimation have the very
significant advantage that they do not depend on a fitting a model of
the data and hence have a degree of robustness missing from parametric
estimators. 

A great deal of theory has been developed on non-parametric entropy
rate estimation from data from finite
alphabets~\cite{kontoyiannis_soukhov1994, shields1992}, since these
models form the basis of discrete codes for communications.  These
estimation techniques have been extended to countably infinite
alphabets~\cite{quas1999entropy, kontoyiannis1998nonparametric}.
However, these approaches cannot be directly applied to
continuous-valued stochastic processes, which are the main topic of
interest here. 

There are approaches that have been designed for estimating the
complexity of continuous-valued stochastic processes:
approximate~\cite{pincus1991approximate},
sample~\cite{richman2000physiological} and permutation
entropy~\cite{bandt2002permutation}. These estimators have been used
in estimating the entropy rate for processes, particularly for those
that are memoryless or have short memory. For example, approximate
entropy has been shown to converge to the entropy rate for independent
and identically distributed processes and first order Markov chains in
the discrete-valued case~\cite{pincus1991approximate}. However, at
best, these estimators are sensitive to their parameter choices; at
worst we shall see that they have severe defects.

In this paper, we compare the existing approaches and develop a new
non-parametric differential entropy rate estimator -- NPD Entropy --
for continuous-valued stochastic processes. It combines the best of
the existing discrete alphabet non-parametric estimators with the
standard signal processing techniques to obtain a estimate that is
more reliable than the alternatives. We have implemented this estimator in Python as a package, npd\_entropy, that is available on GitHub\footnote{https://github.com/afeutrill/npd\_entropy}.


\begin{table*}
	\begin{center}
		\begin{tabular}{l|cc|ccc|ld{3.2}}
                  & \multicolumn{2}{c|}{\textbf{Values}} &
                  \multicolumn{3}{c|}{\textbf{Estimation Quality}} & 
                  &  \multicolumn{1}{r}{\textbf{Computation}} \\
		  \textbf{Estimation Technique} &
                  \textbf{Discrete} & \textbf{Continuous} &
                  \textbf{Consistent} & \textbf{Asymp. unbiased} &
                  \textbf{Correlation Length} & \textbf{Complexity} &
                  \multicolumn{1}{r}{\textbf{Time (s)}} \\ 
			\hline
			Grassberger~\cite{grassberger1989estimating} & \ding{52} & \ding{54} & \ding{54} & \ding{54} & $\approx \log(N)$ &  \\
			Kontoyiannis and Suhov~\cite{kontoyiannis_soukhov1994} & \ding{52} & \ding{54} & \ding{52} &  \ding{52} & $\approx \log(N)$ &  \\
			Approximate Entropy~\cite{pincus1991approximate} & \ding{52} & \ding{52} & \ding{54} & \ding{54} & $m$ & $O(N^2)$  & 10.01\\
			Sample Entropy~\cite{richman2000physiological} & \ding{52} & \ding{52} & \ding{54} & \ding{54} & $m$ & $O(N^2)$  & 263.0\\
			Permutation Entropy~\cite{bandt2002permutation} & \ding{52} & \ding{52} & \ding{54} & \ding{54} & $n$ & $O(n!N)$ & 0.82 \\
			Specific Entropy~\cite{darmon2016specific} & \ding{54} & \ding{52} & \ding{52} & \ding{52} & $p$ & $O(N^2 p)$ & 504,219.9\\
                        NPD Entropy &\ding{54} & \ding{52} & \ding{52} & \ding{52} & $\approx \log(N)$ & $O(N \log N)$ & 39.96 \\ 
		\end{tabular}
		\caption{Comparison of entropy rate estimators. The
                  estimator from Kontoyiannis and Suhov has desirable
                  properties: consistent and asymptotically
                  unbiased.  Approximate, sample and permutation
                  entropy can be applied to either discrete or
                  continuous valued sequences but these are biased and
                  inconsistent. Correlation length refers to the
                  longest lag at which correlations are included into
                  the entropy rate estimate where the length of
                  data is $N$, the length of substrings matched in
                  approximate and sample entropy is $m$, the order of
                  permutations used in permutation entropy is $n$, and
                  $p$ is the number of previous observations used in the conditional entropy calculation. Note
                  also that although specific entropy behaves
                  relatively well, its computation times are
                  prohibitive. We have only included the complexity and computation time for the continuous-valued estimators that we test in this paper. \label{tab:estimator_comparison} \label{tab:run_time_comparison} \label{tab:complexity}} 
	\end{center}
        
\end{table*}

Table~\ref{tab:estimator_comparison} outlines our results. Notably,
some discrete estimators are consistent and asymptotic unbiased, but cannot be applied directly to continuous-valued
processes as they are based on string matching. On the other hand,
the main approaches that have been applied to continuous values --
approximate, sample and permutation entropy -- are not consistent, and use (short) finite windows, limiting their ability to
cope with processes with extended correlations. We show,
for instance, that these estimation techniques do not make accurate
estimates for the entropy rate for processes whose dependency
structure has slowly decaying correlations.

We examine these entropy rate estimates performance on
data generated by long range dependent (LRD) processes. LRD processes
have been shown to be effective models for phenomena such as network
traffic~\cite{Leland:1993:SNE:166237.166255,
  willinger_self_similar_high_variability_1997},
finance~\cite{willinger_stock_1999}, climate
science~\cite{varotsos2006} and hydrology~\cite{hurst1951}.  We apply
the estimation approaches to two common LRD processes with known entropy rate properites. Thus we
can show exactly how bad some estimators are when applied to an even
slightly challenging data set.

Another alternative -- specific entropy -- was developed as a
technique to calculate the predictive uncertainty for a specific state
of a continuous-valued process~\cite{darmon2016specific}. This
approach utilises more rigorous statistical foundations to estimate
the entropy rate of a state given the observation of a finite
past. The technique is able to make accurate entropy rate estimates by
calculating the average over the states and is able to capture the
complex dependency structure with past observations. However this
comes at a large computational cost, and hence cannot be used for
large sequences, or for streaming data to make online estimates (see
Table~\ref{tab:estimator_comparison} for computation time
comparisons).




We develop an estimation technique -- NPD Entropy -- that utilises discrete state non-parametric estimators. We utilise a connection
between the Shannon entropy and differential
entropy~\cite[pg. 248]{cover_thomas_2006}, and then extend it to the
case of entropy rates. The technique
quantises the continuous-valued data into discrete bins, then makes
an estimation of the quantised process using discrete alphabet
techniques. Then the differential entropy rate estimate is calculated
by adjusting by the quantitative difference between the differential
entropy rate and the Shannon entropy rate of the quantised process.

We show that NPD Entropy inherits useful estimation properties
from discrete alphabet estimators. Hence, by
choosing a discrete estimator with the desired
properties we can ensure that NPD Entropy has the same properties for estimation on continuous-valued data. We show that
this performs well in the estimation of differential entropy rate of
stochastic processes which have more complex dependency structure. 

We also compare the runtime performance of techniques and find that
NPD Entropy can make much faster estimates than any approach
of comparable accuracy.


%% file: background_TSP.tex
\section{Background}

\subsection{Entropy Rate}

In this section we'll define the concepts required. Many of these
are standard, but we define them carefully here because of the
proliferation of terms involving ``entropy'' in some manner. First
we'll define the Shannon entropy of a discrete random variable. This
is the classic definition of information entropy from Claude Shannon's
foundational paper~\cite{Shannon48}. 

\begin{definition}\label{shannon_entropy}
	For a discrete random variable, $X$, with support on $\Omega$, and a probability mass function $p(x)$, the Shannon entropy, $H(X)$ is defined as,
	$H(X) = -\sum_{x \in \Omega} p(x)\log p(x).
	$
\end{definition} 

We extend the definition of Shannon entropy for a collection of random variables, called the joint entropy. 

\begin{definition}
	For a collection of discrete random variables, $X_1, ... , X_n$, with support on, $\Omega_1, ... , \Omega_n$ and joint probability mass function $p(x_1, ... , x_n) = p(\mathbf{x})$, we define the joint entropy of the collection of random variables as,
	\begin{align*}
	H(X_1, ... , X_n) &= -\sum_{x_1 \in \Omega_1} ... \sum_{x_n \in \Omega_n} p(\mathbf{x})\log p(\mathbf{x}).
	\end{align*}
\end{definition}

Entropy can be thought of as the average uncertainty or randomness contained in a random variable. We want an entropic notion that can be applied to stochastic processes, and hence we will define the entropy rate, the average uncertainty per random variable.

\begin{definition}
	For a discrete-valued, discrete-time stochastic process, $\raisedchi = \{X_i\}_{i \in \mathbb{N}}$, the entropy rate, is defined where the limit exists as,
	\begin{align*}
	H(\raisedchi) &= \lim\limits_{n \rightarrow \infty} \frac{1}{n} H(X_1, ... , X_n).
	\end{align*}
\end{definition}

In this paper, we are considering the entropy rate of stochastic processes of continuous random variables, so we'll need to extend the notion of entropy to continuous random variables. Claude Shannon in his original treatment of entropy~\cite{Shannon48}, extended the definition by considering the definition of Shannon entropy as the expected value of the information content, \emph{i.e.}, $H(X) = -E[\log(p(X))]$. 

\begin{definition}\label{differential_entropy}
	The differential entropy, $h(X)$ of a random variable, $X$, with support, $\Omega$, and probability density function, $f(x)$, is,
	\begin{align*}
	h(X) &= -\int_{\Omega} f(x)\log f(x) dx.
	\end{align*}
\end{definition}

Differential entropy has some important properties which are different
from Shannon entropy. For example, differential entropy can be
negative, or even diverge to $-\infty$, which we can see by
considering the Dirac delta function, $\delta(x)$, \emph{i.e.}, the
unit impulse, defined by the properties $\delta(x) =0$ for $x \neq
0$ and $\int_{-\infty}^{\infty} \delta(x) \, dx = 1$. The Dirac delta
can be thought of in terms of probability as a completely determined
point in time, that is, a function possessing no uncertainty. It can
be constructed as the limit of rectangular pulses of constant area 1
as their width decreases, equivalent to the density of a uniform random variable, and hence we can calculate the entropy of
the Dirac delta as
\begin{align*}
h(X) = -\int_{-a}^{a} \frac{1}{2a} \log\left(\frac{1}{2a}\right) dx
= \log(2a),
\end{align*}
which tends to $-\infty$ as  $a \rightarrow 0$.


Similar to Shannon entropy we define the joint differential entropy for a collection of continuous random variables. 

\begin{definition}
	The joint differential entropy of a collection of continuous random variables, $X_1, ... , X_n$, with support on, $\Omega_1 \times ... \times \Omega_n$, with a joint density function, $f(x_1, ... , x_n) = f(\mathbf{x})$, is
	\begin{align*}
	h(X_1, ... , X_n) &= -\int_{\Omega_1} ...\int_{\Omega_n} f(\mathbf{x}) \log f(\mathbf{x}) d\mathbf{x}.
	\end{align*}
\end{definition}

Now we can define the object that we want to estimate, the differential entropy rate. This is a direct analogue of the discrete-valued random variable case.

\begin{definition}
	The differential entropy rate for a continuous-valued, discrete-time stochastic process, $\raisedchi = \{X_i\}_{i \in \mathbb{N}}$, is defined where the limit exists as,
	\begin{align*}
	h(\raisedchi) &= \lim\limits_{n \rightarrow \infty} \frac{1}{n} h(X_1, ... , X_n).
	\end{align*}
\end{definition}

We state an equivalent characterisation of the differential entropy rate for stationary stochastic processes, using conditional entropy. This was developed for the Shannon entropy of discrete processes, however can be extended to differential entropy utilising an identical argument~\cite[pg. 416]{cover_thomas_2006}.

\begin{theorem}\label{conditional_entropy_rate}
	For a stationary stochastic process, $\{X_i\}_{i \in \mathbb{N}}$, the differential entropy rate where the limit exists is equal to,
	\begin{align*}
	h(\raisedchi) &= \lim\limits_{n \rightarrow \infty} h(X_n | X_{n-1}, ... , X_1).
	\end{align*}
\end{theorem}

\noindent This states that the differential entropy rate is the limit
of the new information that we get from each new random variable,
after observing the infinite past.

\subsection{Strongly Correlated Processes}
Now we will define the notion of a stochastic process with strong
correlations, which we will use to test the effectiveness of the
entropy rate estimation techniques.  We will be using processes which
are called long range dependent (LRD). In many contexts, {\em e.g.,}
\cite{Leland:1993:SNE:166237.166255,willinger_self_similar_high_variability_1997,willinger_stock_1999,varotsos2006,hurst1951},
these processes are the norm, not pathologies. LRD processes have an
autocorrelation function, $\rho(k)$, which decays as a power law in
contrast to short-range correlated processes whose autocorrelations
decay exponentially, and which have typically been used to test
entropy estimators. However, LRD processes are still stationary, and
can be modelled relatively simply, and thus form a practically useful
test case, which should be seen as only mildly challenging. 

\begin{definition}
	Let $\{X_n\}_{n \in \mathbb{N}}$ be a stationary process. If there exists an $\alpha \in (0,1)$ and a constant $c_\rho > 0$, such that,
	\begin{align*}
	\lim\limits_{k \rightarrow \infty} \frac{\rho(k)}{c_\rho k ^{-\alpha}} = 1.
	\end{align*}
	Then we say that, $X_n$, is long range dependent.
\end{definition}

An implication of this is that the sum of the autocorrelations
diverges, {\em i.e.,} $\sum_{k=1}^{\infty} c_\rho k ^{-\alpha} =
\infty$, which is often given as an equivalent definition of
LRD~\cite[pg. 43]{beran1994statistics}. 

A parameter of interest when analysing the strength of correlations is
the Hurst parameter, $H$. It it linked to the exponent of the power-law decay by the relationship~\cite[pg. 42]{beran1994statistics}, $H = 1 - \alpha/2$.
$H$ takes values
between 0 and 1, with $H > \frac{1}{2}$ representing the region where
the process is positively correlated, with the process becoming
completely correlated at $H = 1$. We call processes in this case LRD, otherwise we call them short range dependent (SRD). The region $H < \frac{1}{2}$
represents the region where a process is negatively correlated, with a
completely negatively correlated process is at $H=0$. The autocorrelation functions sum to one in this case,
$\sum_{k=1}^{\infty} \gamma(k) = 0$, and is called constrained short
range dependent (CSRD)~\cite{gefferth2003nature}. This case much less studied,
and the differential entropy rate
differs between FGN and ARFIMA(0,d,0)~\cite{feutrill2021differential}. Finally the mid
point, $H=\frac{1}{2}$ are processes with short range or no correlations,
\emph{i.e.}, white Gaussian noise. 

\subsubsection{Fractional Gaussian Noise}

Fractional Brownian Motion (FBM) is a continuous time Gaussian
process, $B_H(t)$, on $[0, T]$, $B_H(0) = 0$, $E[B_H(t)] = 0,
\forall t \in [0,T]$, and with covariance function, 
\begin{align*}
E[B_H(t)B_H(s)] &= \frac{1}{2} \left( t^{2H} + s^{2H} - |t-s|^{2H} \right),
\end{align*}

\noindent where $H$ is the Hurst
parameter~\cite{MandelbrotNess1968}. Fractional Gaussian (FGN) Noise
is the stationary increment process of FBM, {\em i.e.,} $X(t) = B_H(t+1) - B_H(t)$.

The entropy rate of FGN is~\cite{feutrill2021differential}
\begin{align*}
	h(\raisedchi) = \frac{1}{2} \log (2 \pi e) + \frac{1}{4 \pi} \int_{-\pi}^{\pi} \log f(\lambda) d\lambda,
\end{align*}
 with the spectral density 
 \begin{align*}
 	f(\lambda) = 2 c_f (1 - \cos\lambda) \sum_{j=-\infty}^{\infty} |2 \pi j + \lambda |^{-2H -1},
 \end{align*}
 where $c_f = \frac{\sigma^2}{2\pi}\sin\left(\pi H\right) \Gamma\left(2H+1\right)$. Here we
 calculate the integral numerically.

\begin{figure}
	\centering
	\includegraphics[width=0.8\linewidth]{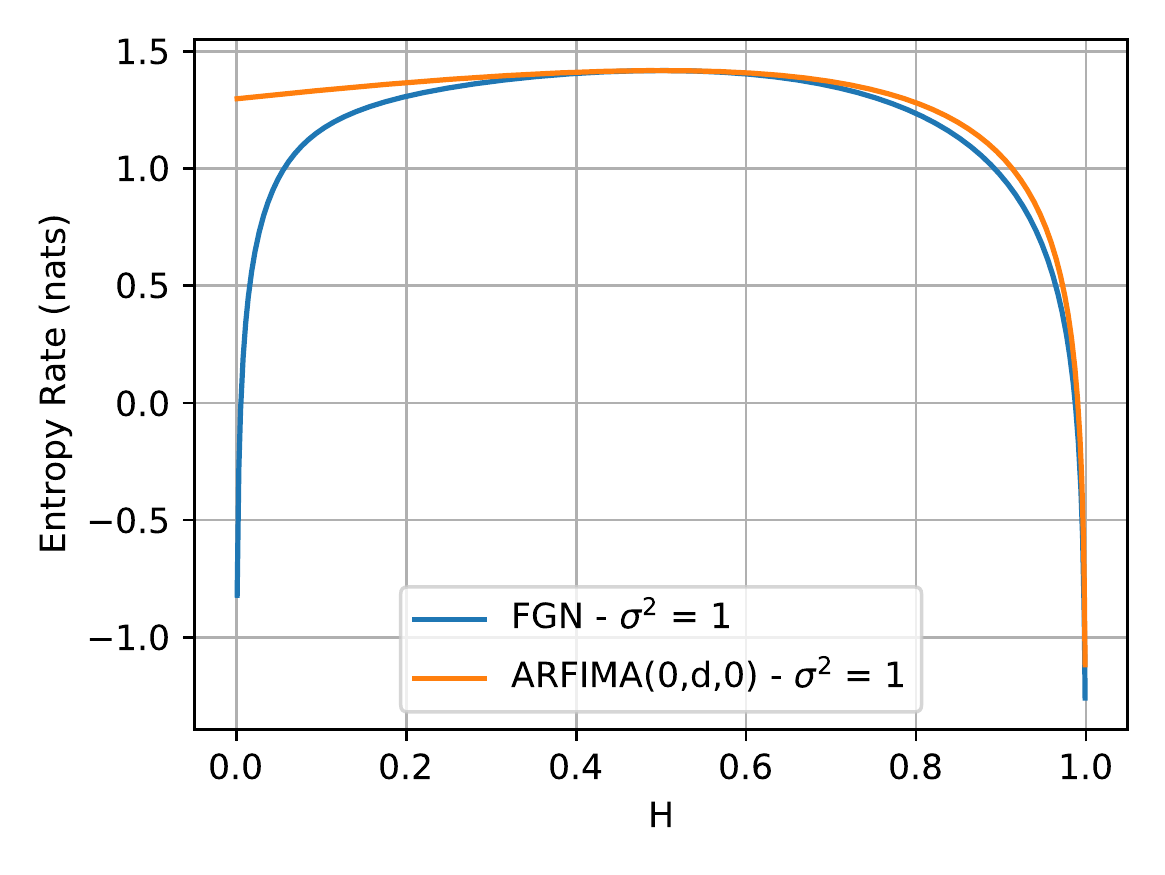}
	\caption{The entropy rate of FGN and ARFIMA(0,d,0) processes with variance, $\sigma^2 = 1$, as function of the Hurst parameter. The maximum for both processes  is at $H=0.5$, and the function tends to $-\infty$ asymptotically as $H \rightarrow$ 0 and 1 for FGN and $H \rightarrow$ 1 only for ARFIMA(0,d,0)\cite{feutrill2021differential}.}
	\label{fig: fgn_entropy_rate}
\end{figure}

We plot the entropy rate of FGN as a function of the Hurst parameter
in Figure~\ref{fig: fgn_entropy_rate}. The entropy rate function
attains a maximum at $H=0.5$, which is where the process is completely
uncorrelated. As the process tends
to 0 or 1, increasing the strength of the positive or negative
correlations, the entropy rate function decreases. The stronger the
correlations, more information is known and hence
there is lower uncertainty about future values.

\subsubsection{ARFIMA(0,d,0)}

An ARFIMA(0,d,0) process, $\{X_n\}_{n \in \mathbb{N}}$, is the
stationary solution of the following equation for $-\sfrac{1}{2} < d <
\sfrac{1}{2}$ (the parameters $d$ and $H$ are related as $d = H -
\sfrac{1}{2}$) 
\begin{align*}
(1-L)^{d} X_n &= \epsilon_n,
\end{align*}
where $L$ is the lag operator, which is defined as $L X_n = X_{n-1}$~\cite{hosking1981}.
 The autocorrelation function is equal to,
 \begin{align*}
 	 \rho(k) = \frac{\Gamma(1-d)\Gamma(k+d)}{\Gamma(d)\Gamma(k+1-d)}. 
 \end{align*}
 As $\Gamma(k+a)/\Gamma(k+b)$ tends to $k^{a-b}$ asymptotically, 
 \begin{align*}
 \rho(k) \sim \frac{\Gamma(1-d)}{\Gamma(d)} k^{2d-1}, \text{ as}, k \rightarrow \infty.
\end{align*}
For $0 < d < \sfrac{1}{2}$ the sum of autocorrelations diverges, and
hence ARFIMA(0,d,0) is a positively correlated LRD process.  For
$-\sfrac{1}{2} < d < 0$ the sum of autocorrelations converges and
$\sum_{k=0}^{\infty} \rho(k) = 0$, and hence the process is CSRD. If
$d=0$ the process is the white noise process $X_n = \epsilon_n$.

The entropy rate of an ARFIMA(0,d,0) process as a function of $H$ is~\cite{feutrill2021differential}
\begin{align*}
	h(\raisedchi) = \frac{1}{2} \log\left(2 \pi e \sigma^2\right) +
        \log&\left(\Gamma\Big(\frac{3}{2} - H\Big)\right)\\
         - &\frac{1}{2} \log\big(\Gamma(2 - 2H)\big),
\end{align*} 
for process variance, $\sigma^2$. Figure~\ref{fig: fgn_entropy_rate}
shows the resulting entropy rate as a function of $H$. For the LRD region
FGN and ARFIMA(0,d,0) have very similar properties, however in the
CSRD region ARFIMA(0,d,0) doesn't diverge to $-\infty$. Hence, even
though the two processes, FGN and ARFIMA(0,d,0) are often considered
as simple alternative models with respect to their autocorrelation,
they have different entropy rate behaviour.  This differing entropy
rate behaviour is one feature we will use to test the effectiveness of
estimators.

%% file: discrete_time_estimators_TSP.tex
\section{Discrete-valued, Discrete-Time Estimation}\label{discrete_entropy}
In this section we introduce a non-parametric estimator based on the Lempel-Ziv compression
algorithm~\cite{ziv1977universal}, which in future sections we can use as the basis of an estimator in the continuous-valued case. The estimation technique is based
on a limit theorem on the frequency of string matches of a given
length. For every position $i$ in a sequence ${\mathbf x}$, and
segment length $\ell \geq 1$ we examine the segment $x_{i}^{i+\ell-1}
= (x_i, x_{i+1}, \ldots, x_{i+\ell -1})$. For each $i$, and a given
window length $n \geq 1$ we we find the longest segment starting from
$i$ that also appears in the window. Equivalently, we can look for the
minimum length of the non-matching sequence as defined below: 
\begin{align*}
	L^n_i({\mathbf x}) = \min\{L: x_i^{i+L-1} \neq x_j^{j+L-1}, 1 \leq j \leq n, j \neq i\},
\end{align*} 
that is, the length of the shortest prefix of $x_i, x_{i+1} ,\ldots$
which is not a prefix of any other $x_j, x_{j+1}, ... $ for $j \leq
n$. A limit theorem was developed by Wyner and Ziv~\cite{wyner1989},
based on string matching, which states
\begin{align*}
\lim\limits_{n \rightarrow \infty} \frac{L^n_i(x)}{\log n} \rightarrow \frac{1}{h(\raisedchi)}, \text{ in probability}.
\end{align*}
This can be turned into an estimation technique by truncating the sequence for some finite $n$, and calculating the reciprocal of $\frac{L^n_i(x)}{\log n}$. 
      
\noindent This was extended to almost sure convergence, a.s., by
Ornstein and Weiss~\cite{ornstein1993entropy}, meaning that prefix
sequences that do not converge have probability 0. These approaches, however, suffer from the issue of basing the estimation on one long sequence of data. Utilising the idea
of the theorem above, estimation techniques were developed which utilise
multiple substrings and average the $L^n_i$'s instead of estimating
from one long string. The
following statement, by Grassberger~\cite{grassberger1989estimating}, that the estimator is consistent, was suggested as almost sure convergence heuristically,
\begin{align}\label{grassberger_limit}
\lim\limits_{n \rightarrow \infty} \frac{\sum_{i=1}^{n} L_i^n(x)}{n \log n} = \frac{1}{h(\raisedchi)}.
\end{align}

\noindent This expression was shown not to be true in general by
Shields~\cite{shields1992}, except in the cases of independent and
identically distributed (i.i.d.) processes and Markov chains, with almost sure convergence. However,
a weaker version does hold for general ergodic processes, which states
that for a given $\epsilon > 0$, that all but a fraction of size at
most $\epsilon$ of the ${\sum_{i=1}^{n} L_i^n(x)}/{n \log n}$, are
within the same $\epsilon$ of $1/h(\raisedchi)$, in the limit as $n \rightarrow \infty$~\cite{shields1992}.

This is converted to an estimation technique by taking a suitably
large $n$ and the reciprocal of the above expression for $\hat{h}$.
However, we aim to make consistent estimates for processes with strong correlations where the statement from Grassberger, (\ref{grassberger_limit}) doesn't hold, hence we require estimators that are consistent for a wider variety of correlation structures. The estimator remains biased, even taking the asymptotic limits. Kontoyiannis and
Suhov~\cite{kontoyiannis_soukhov1994}, and Quas~\cite{quas1999entropy}
extended this approach to a wider range of processes on a discrete alphabet $\mathcal{A}$, firstly
stationary ergodic processes on a finite alphabet that obey a Doeblin condition,
\emph{i.e.}, there exists an integer $r \ge 1$ and a real number
$\beta \in (0,1)$ such that for all $x_0 \in \mathcal{A}, Pr(X_0 = x_0 |
X^{-r}_{-\infty}) \le \beta$, with probability one, and secondly to
processes with infinite alphabets and to random fields satisfying the
Doeblin condition. The estimator from these limit theorems is then consistent and becomes unbiased in the asymptotic limit.
  
In subsequent sections we will quantise continuous valued processes to discrete processes, and then make estimates of the Shannon entropy rate, and then convert these to differential entropy rate estimates. Hence, the properties that have been shown, such as the consistency and bias, are properties that we will be looking to preserve to make reliable non-parametric estimates.

%% file: continuous_time_estimators_TSP.tex
\section{Continuous valued, discrete time Estimation}

We consider some non-parametric estimators of entropy rate for continuous-valued data in two different classes, relative measures, that we can use for comparison of complexity of a system, and absolute measures, which are intended to accurately estimate the value of differential entropy rate for a system.

These techniques have all been developed to quantify the complexity of continuous-valued time series, and hence the intention is to compare time series as opposed to provide an absolute estimate. These types of measures, from dynamic systems literature, and have been successful in the analysis of signals to detect change, in a variety of applied contexts~\cite{alcaraz2010review, chen2006comparison, lake2002sample}. 

The final technique we consider, Specific Entropy~\cite{darmon2016specific}, is an absolute measure of the entropy rate. Due to computational advances, the technique uses non-parametric kernel density estimation of the conditional probability density function, based on a finite past, and uses this in a ``plug in" estimator.

\subsection{Approximate Entropy}\label{approx_ent}

Approximate entropy was introduced in Pincus~\cite{pincus1991approximate}, with the intention of classifying complex systems, however it has been used to make entropy rate estimates, due to the connection with entropy rate for finite Markov chains and i.i.d. processes. Given a sequence of data, $x_1, x_2, ..., x_N$, we have parameters $m$ and $r$, which represent the length of the substrings we will define and the maximum distance, according to a distance metric, between substrings to be considered a match. Then we create a sequence of substrings, $u_1 = [x_1, ... , x_m], u_2 = [x_2, ... , x_{m+1}], ... , u_{N- m + 1} = [x_{N-m + 1}, ... , x_N]$. First we define a quantity, 
\begin{align*}
C^m_i(r) &= \frac{1}{N-m+1} \sum_{j=1}^{N-m+1} \mathbbm{1}_{\{d[u_i, u_j] \le r\}},
\end{align*}

\noindent where $d[x(i), x(j)]$ is a distance metric. Commonly used metrics for this measure are the $l_\infty$ and $l_2$ distances,
\begin{align*}
	l_\infty(x_i, x_j) &= \max_{k = 0, ... , m-1} |x_{i + k} - x_{j + k}|, \text{ or},\\
	l_2(x_i, x_j) &= \sqrt{(x_i - x_j)^2 + ... + (x_{i+m-1} - x_{j+m-1})^2}.
\end{align*}

The following quantity, used in the calculation of the Approximate Entropy, is defined in Eckmann and Ruelle~\cite{eckmann1985ergodic}, 
\begin{align*}
\Phi^m(r) &= \frac{1}{N-m+1} \sum_{i=1}^{N-m+1} \log C^m_i(r).
\end{align*}

Then we define Approximate entropy, $ApEn(m, r)$ as
\begin{align*}
ApEn(m, r) &= \lim\limits_{N \rightarrow \infty} [\Phi^m(r) - \Phi^{m+1}(r)].
\end{align*}

Therefore each estimate of approximate entropy is with respect to a finite amount of data, of length $N$, and not the true value for a system. This is a biased statistic, which arises from the the calculation of $C^m_i(r)$ quantities, where a substring is counted twice, and the application of the logarithm function, which is concave, contributing to the bias, since $E[\log(X)] \le \log(E[X])$ by Jensen's inequality~\cite{delgado-bonal2019}. The bias in this estimator decreases as the number of samples, $N$, gets larger~\cite{delgado-bonal2019}, and this has been shown to bias towards having more complexity in a time series~\cite{richman2000physiological}.

Pincus showed in his initial paper, that the Approximate Entropy would converge to the entropy rate for i.i.d. and finite Markov chains~\cite{pincus1991approximate}. However, our intention is to make accurate estimates for stochastic processes with a more complex correlation structure, which we provide in Section~\ref{testing_measures}.
Approximate entropy is quite sensitive to the two parameters, $m$, and $r$, and hence care must be taken when selecting these parameters~\cite{delgado-bonal2019, yentes2013appropriate}. From previous investigations, it is recommended that $m$ has a relatively low value, \emph{i.e.}, 2 or 3, which will ensure that the conditional probabilities can be estimated reasonably well. The recommended values for $r$, are in the range of $0.1\sigma - 0.25\sigma$, where $\sigma$ is the standard deviation~\cite{delgado-bonal2019}. In this paper, we will be using the values of $m = 3$ and $r=0.2$, as we use the variance $\sigma^2 = 1$, for all processes. 

\subsection{Sample Entropy}

A closely related technique for estimating the entropy rate is sample entropy~\cite{richman2000physiological}, which was developed to address some issues in approximate entropy. It is highly dependent on the data length of the time series and the lack of relative consistency, \emph{i.e.}, that the approximate entropy of a time series is consistently smaller than another for all values of $r$, which is an important property which will give us confidence in our complexity measurements. The sample entropy is a simpler algorithm than approximate entropy, that is quicker to make an estimate and eliminating self-matches in the data.

We now define sample entropy, by using very similar objects to approximate entropy. Given a time series, $x_1, ... , x_N$, of length $N$, we again calculate substrings $u^m_i = [x_i, ... , x_{i+m-1}]$ of length $m$, and the parameter, $r$, for the maximum threshold between strings for comparison. We calculate
\begin{align*}
A = \sum_{i=0}^{N-m} \mathbbm{1}_{\{d[u^{m+1}_i, u^{m+1}_j] < r\}}, B = \sum_{i=0}^{N-m+1} \mathbbm{1}_{\{d[u^m_i, u^m_j] < r\}},
\end{align*}

\noindent where $d[u^m_i, u^m_j]$ is a distance metric. Finally we define SampEn as,
\begin{align*}
SampEn &= -\log (A/B).
\end{align*}

As $A$ will be always less than or equal to $B$, this value will always be non-negative.

Sample entropy removes the bias that is introduced via the extra counting of vectors in the formation of the statistic, however it doesn't reduce the source of bias that is introduced by the correlation of the substrings~\cite{richman2000physiological, delgado-bonal2019}.

The same parameter selection issues as approximate entropy also apply~\cite{delgado-bonal2019}, and hence we choose $m=3$ and $r=0.2$.

\subsection{Permutation Entropy}

Permutation entropy is a complexity measure of time series that uses order statistics to estimate and classify the complexity of observed time series. This is a relative measure, however we expect it to pick up the trend of the entropy rate function. The technique was first defined by Bandt and Pompe~\cite{bandt2002permutation}, and subsequent investigations into order statistics which further develops the theory that was used in permutation entropy~\cite{bandt2007order}.

Now we will define Permutation entropy, first by defining the order statistics required. Given a time series, $x_1, ... , x_N$, which we assume to be stationary, we will be considering permutations, $\pi \in \Pi$ of the substrings of length, $n$, of which there are $n!$ different permutations. For each permutation, $\pi$, we define the relative frequency as,
\begin{align*}
p(\pi) &= \frac{|\{t | t \le N - n, x_{t+1}, ... , x_{t+n}\text{ has type }\pi|\}}{N - n +1}.
\end{align*}

\noindent Hence, we are working with approximations to the real probabilities, however we could recover these by taking the limit as $N \rightarrow \infty$ by the Law of Large Numbers~\cite[pg. 73]{Durrett2010} using a characteristic function on the permutation, with a condition on the stationarity of the stochastic process. 

The permutation entropy of a time series, of order $n \ge 2$, is then defined as,
\begin{align*}
H(n) &= -\sum_{\pi \in \Pi} p(\pi)\log p(\pi).
\end{align*}

Permutation entropy has one parameter, the order $n$. The number of permutations for an order scales as $n!$, which provides a time complexity issue as the require computations grows very quickly in the size of the order. Hence, the minimum possible data required to observe all of the possible permutations of order $n$, is $n!$ data. However, it is claimed that the Permutation Entropy is robust to the order of the permutations used~\cite{bandt2002permutation}. To lower the complexity and provide robust measurements we will use $n=3$ in this paper.

\subsection{Specific Entropy}

Specific entropy was defined by Darmon~\cite{darmon2016specific}, to provide a differential entropy rate estimation technique that is on a stronger statistical footing than the previously defined estimation techniques.  The approach is to consider the short-term predictability of a sequence, by utilising a finite history of values to create a kernel density estimate of the probability density function, then use the kernel density estimate to plug in to the differential entropy rate formula in Definition~\ref{differential_entropy}. For the calculation of this quantity, we define a parameter, $p$, which is the length of the history that is considered in the kernel density estimation of the conditional probability density function. 

The definition of the specific entropy rate, makes a finite truncation of the conditional entropy version of the entropy rate, from Theorem~\ref{conditional_entropy_rate}. This assumes stationarity of the conditional distributions used, up until $p$ observations in the past. Darmon~\cite{darmon2016specific}, shows that the conditional entropy up to order $p$, depends on the state specific entropy of a particular past $(x_p, \ldots, x_1) = \mathbf{x^p_1}$ and the density of the possible pasts $(X_p, \ldots, X_1) = \mathbf{X^p_1}$~\cite{darmon2016specific}. This is shown by an argument which establishes that, 
\begin{align*}
	h\left(X_t | \mathbf{X^{t-1}_{t-p}}\right) = - E\left[E\left[\log f \left(X_t | \mathbf{X^{t-1}_{t-p}}\right)\right]\right].
\end{align*}

Given this relationship and the law of total expectation, they define the specific entropy rate, of order $p$, $h_t^{(p)}$, as
\begin{align*}
	h_t^{(p)} &= h\left(X_t | \mathbf{X^{t-1}_{t-p}} = \mathbf{x^{t-1}_{t-p}}\right),\\
	&= -\int_{-\infty}^{\infty} f\left(x_{p+1} | \mathbf{x^p_1}\right) \log f\left(x_{p+1} | \mathbf{x^p_1}\right) dx_{p+1}.
\end{align*}

Hence, the specific entropy rate estimator, $\hat{h}_t^{(p)}$, is defined by plugging in the estimate of the density obtained by kernel density estimation, $\hat{f}\left(x_{p+1} | \mathbf{x^p_1}\right)$, is
\begin{align*}
	\hat{h}_t^{(p)} = -E\left[\log \hat{f}\left(x_{p+1} | \mathbf{x^p_1}\right)\right].
\end{align*}

\noindent Then the estimate of the differential entropy rate of order $p$, $\hat{h^{(p)}}$, is defined as
\begin{align*}
	\hat{h}^{(p)} &= \frac{1}{T-p} \sum_{t = p}^{T} \hat{h}_t^{(p)},
\end{align*}
which is the time average of all the specific entropy rates across observed states.


Specific entropy relies on some parameters to construct the kernel density estimation, which are the length of the past, $p$ and the $p+1$ bandwidths, $k_1, \ldots, k_{p+1}$ that are used in the kernel density estimation~\cite{darmon2016specific}. The parameter choice can have large impacts on the quality of the estimation. The suggested technique for selecting $p$ is a cross-validation technique which removes an individual observation and $l$ observations either side. A suggested approach is to take $l=0$ and only remove the individual observation~\cite{darmon2018information}. In practice, it is advised to fix a $p$ and then calculate the bandwidths due to the computational complexity of the cross-validation~\cite{darmon2016specific}. In this work to balance the capturing of the strong correlations and the complexity we choose $p=10$.

\section{Accuracy of Continuous-Value Estimators}\label{testing_measures}
 
In this section we will analyse samples of continuous valued, discrete
time long range dependent data using the existing continuous-value
entropy rate estimators. Most analysis and tests of Approximate and
Sample entropy have been based on i.i.d. processes or finite state
Markov chains \ie processes with either very short or no
correlations.  However, many real processes, in particular the types of
processes for which complexity measures are useful, exhibit long-range
correlations. To create sample data for testing we created 50 samples of 2000 data points, and averaged the estimates to get the values presented. Figure~\ref{fig: comparison_approx_sample_perm} shows
just how bad common estimators are for FGN. Both the shape and scale
of the entropy estimate curves are quite wrong (with the exception of
the shape of the sample entropy).

\begin{figure}[t]
  \centering
	\includegraphics[width=0.9\linewidth]{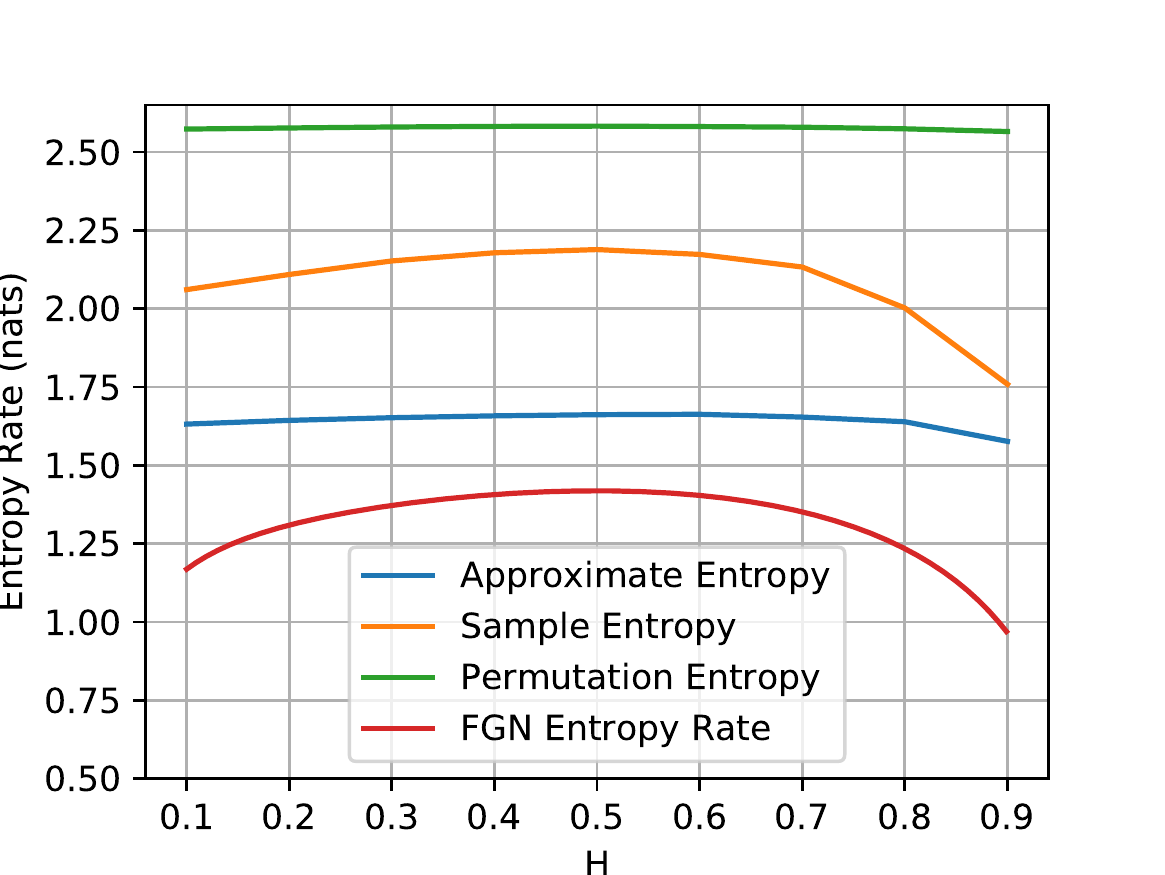}
	\caption{Approximate, sample and permutation entropy estimates
          for FGN. Note the wide discrepencies.}
	\label{fig: comparison_approx_sample_perm}
\end{figure}

Figures~\ref{fig: approx_entropy} and~\ref{fig: sample_entropy} show
the entropy rate estimates for the Sample and Approximate Entropy for
both Fractional Gaussian Noise and ARFIMA(0,d,0) for two different
parameters $m=2,3$ and $r=0.2$ as recommended in Delgado-Bonal and
Marshak~\cite{delgado-bonal2019}.

Sample Entropy approximates the the shape of the real entropy rate
functions but provides large overestimates. This reinforces the use of
Sample Entropy as a relative measure of complexity of a time series,
as long as it is not used generally for the estimation of differential
entropy rate.

Approximate Entropy, however, fails to even approximate trend or range
of values. Interestingly, this technique is also quite sensitive to
changes in the value of $m$.

Figure~\ref{fig: permutation} shows the Permutation Entropy of order
$n=3$. These results are indicate that permutation entropy is not a
good choice for strongly correlated process as all the behaviour
exists on a very small scale, with small differences across the range of $H$. The maximum
Permutation Entropy for $n=3$, is $H(\pi) = \log_2(3!) \approxeq
2.585$, and all of the estimates are within 0.2 of the maximum
value. Similar to Approximate and Sample
entropy, from Figure~\ref{fig: comparison_approx_sample_perm} although
the trend looks promising on a small scale, once we compare the
behaviour to the actual entropy rate, it is a poor estimator of the
actual entropy rate.

\newlength{\thisfigwidth}
\setlength{\thisfigwidth}{0.65\linewidth}

\begin{figure*}[t]
  \centering
  \begin{subfigure}[t]{\thisfigwidth}
	\includegraphics[width=\thisfigwidth]{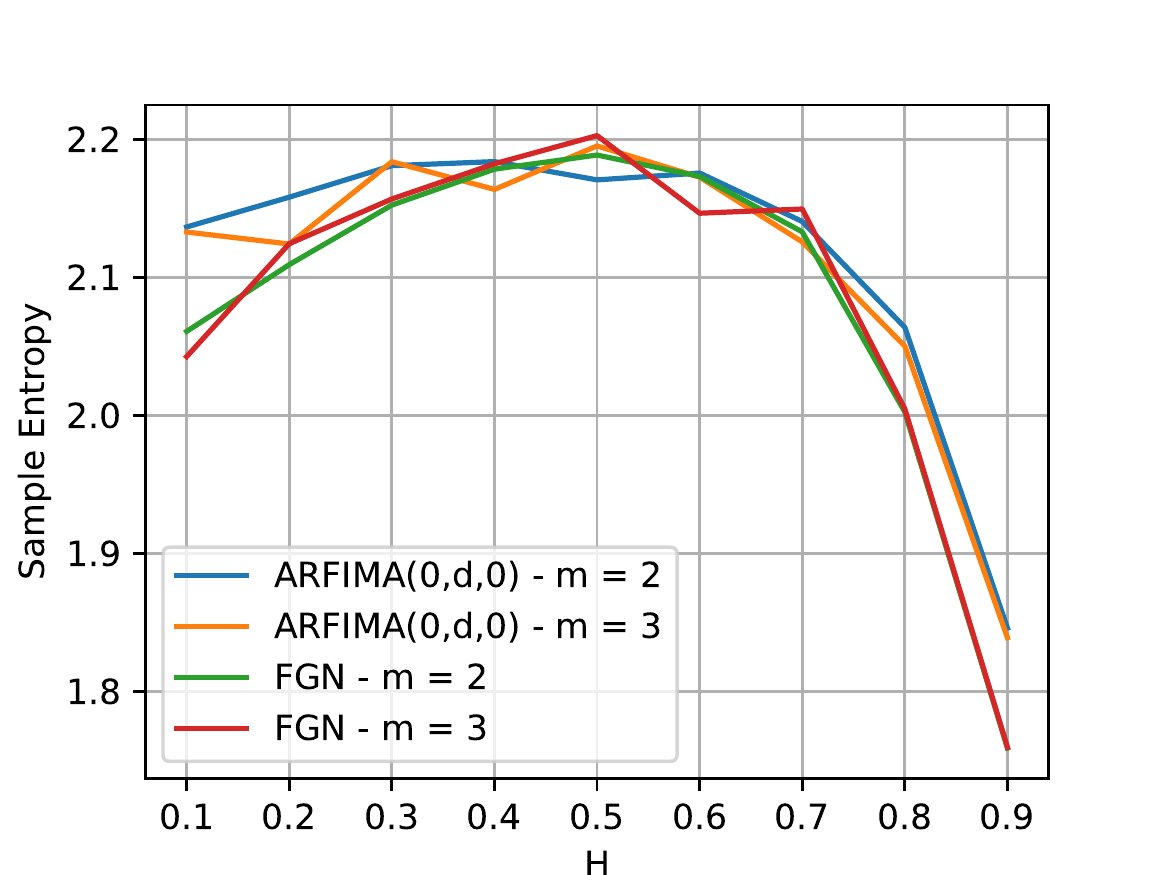}
	\caption{Sample Entropy ($r=0.2$) shows the rough trend of the
          LRD processes, however it overestimates the differential
          entropy by a large amount.}
	\label{fig: sample_entropy}
  \end{subfigure}
  \hfill
  \begin{subfigure}[t]{\thisfigwidth}
    \includegraphics[width=\textwidth]{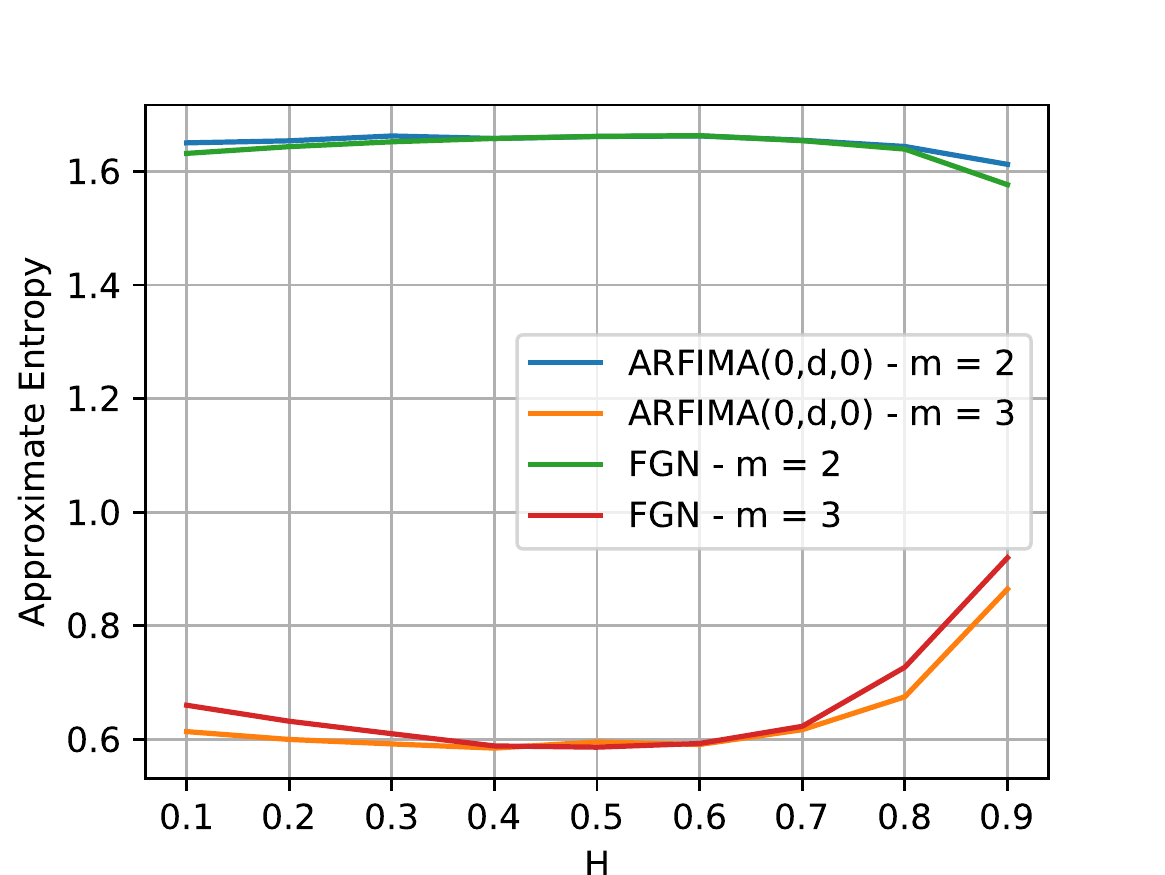}
	\caption{Approximate Entropy ($r=0.2$) fails to even approximate trend
          or range of values. Interestingly, this technique is also quite sensitive
          to changes in the value of $m$.} 
	\label{fig: approx_entropy}
  \end{subfigure}
  \hfill
  \begin{subfigure}[t]{\thisfigwidth}
	\centering
	\includegraphics[width=0.98\linewidth]{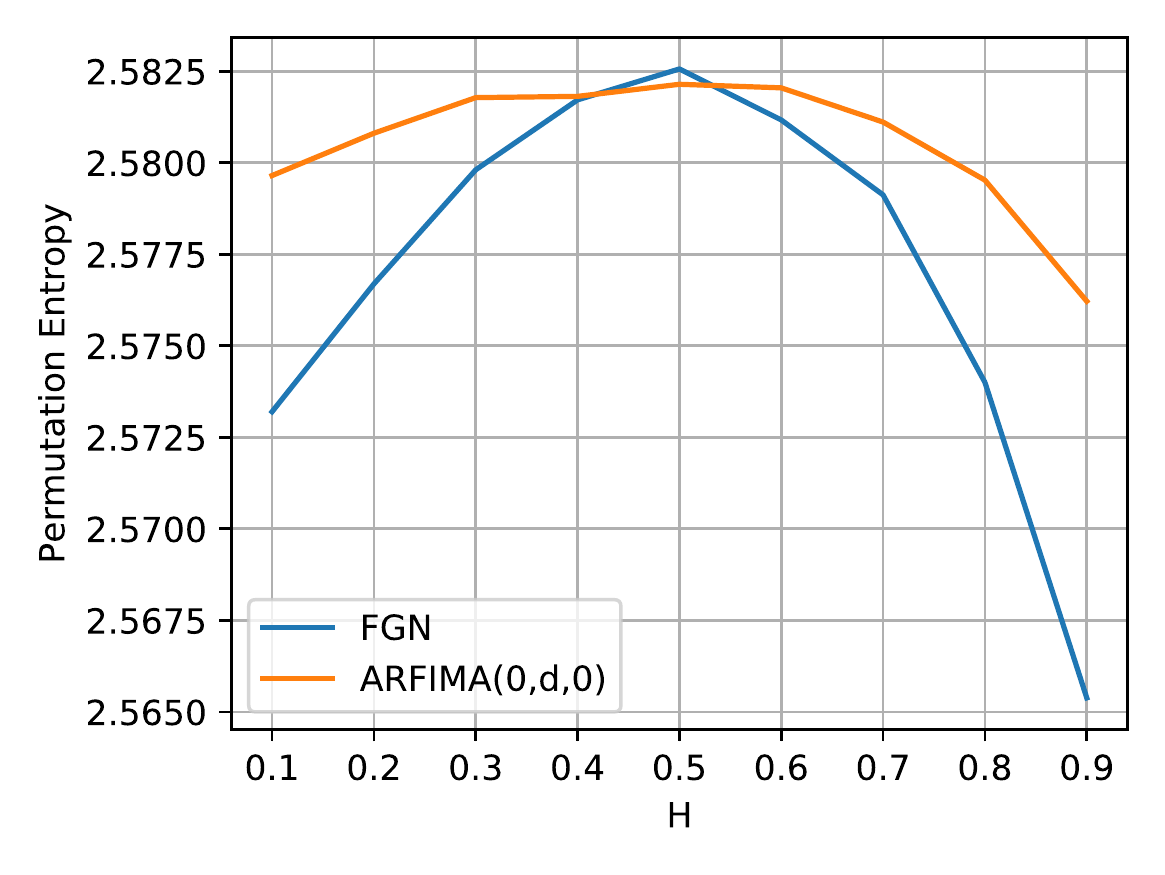}
	\caption{Permutation Entropy (order 3) appears to capture the
          trend of the entropy rate function, however the scale of the
          changes is wrong. Noting that the maximum Permutation
          Entropy is $\log_2(3!) \approx 2.585$ the range displayed is
          very small.}
	\label{fig: permutation}
  \end{subfigure}
  \caption{Detailed entropy-rate estimates of FGN and ARFIMA(0,d,0)
    with process variance $\sigma^2 = 1$.}
\end{figure*}

Naively, these estimates can be improved by extending the lengths of
the windows being used, however, for LRD processes large windows would
be required. Unfortunately the computational complexity of the
approaches (see Table~\ref{tab:estimator_comparison}) grows with
increased window sizes and we see stability problems at least with
Approximate Entropy, so there does not appear to be a suitable
trade-off between computational cost and accuracy. 

Estimates derived from Specific Entropy are shown in Figures~\ref{fig:
  specific_fgn} and~\ref{fig: specific_arfima}, with $p=10$. Specific
entropy provides good agreement with the entropy rate for LRD FGN and
ARFIMA(0,d,0), \emph{i.e.},$H > \frac{1}{2}$, with greater divergence
occuring in the CSRD parameter range for small $H$ values. However,
note that the computational cost for these estimates is very high. 


 
\newlength{\thisfigwidthb}
\setlength{\thisfigwidthb}{0.85\linewidth}

\begin{figure*}[t]
  \begin{subfigure}[t]{\thisfigwidthb}
	\centering
	\includegraphics[width=\thisfigwidthb]{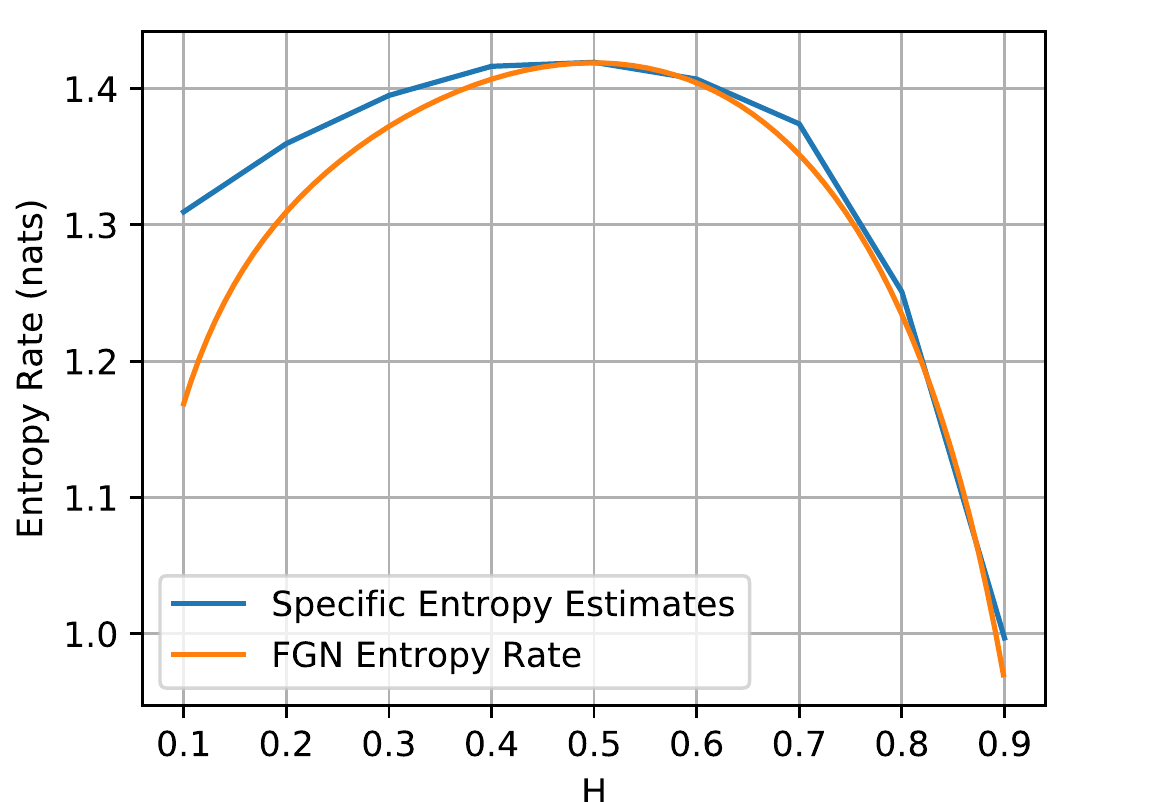}
	\caption{FGN.} 
	\label{fig: specific_fgn}
  \end{subfigure} 
  \hfill
  \begin{subfigure}[t]{\thisfigwidthb}
	\centering
	\includegraphics[width=0.97\linewidth]{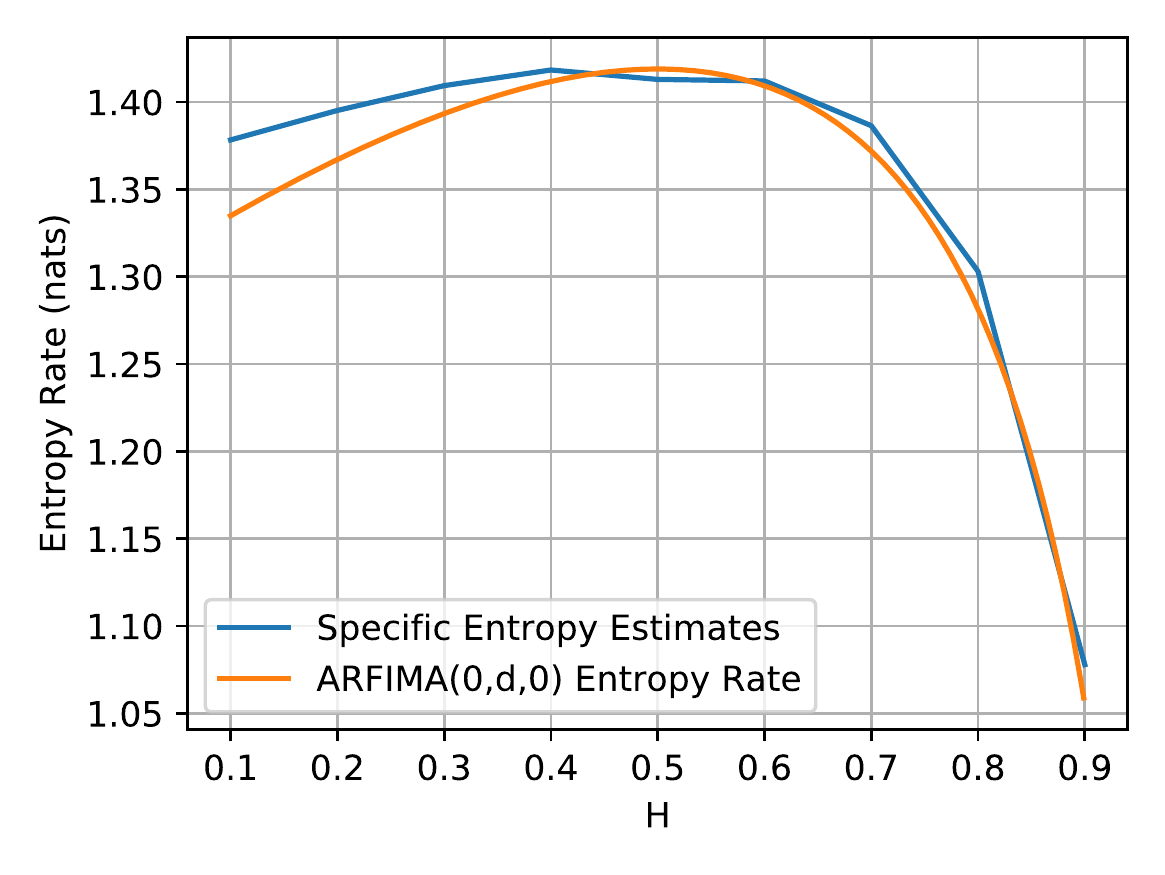}
	\caption{ARFIMA(0,d,0).} 
	\label{fig: specific_arfima}
  \end{subfigure}
  \caption{Specific Entropy ($p=10$). These estimates
          are very good over the range of $H$, however the estimates
          start to diverge for smaller values of $H$.}
\end{figure*}

%% file: shannon_differential_link_TSP.tex
\section{Shannon and Differential Entropy}

In this section, we make a connection between the differential and Shannon entropy rate,
which we use as the basis of an estimation technique. We propose an
approach where we quantise the continuous valued process into a
discrete valued process, for a defined quantisation window size,
$\Delta$, then apply estimators from the discrete valued domain, and
then translate back to differential entropy.

Given the definitions of Shannon and Differential entropy, Definitions~\ref{shannon_entropy} and~\ref{differential_entropy} respectively, we utilise a connection that exists between these quantities for a quantised version of continuous data, which is given by the following theorem~\cite[pg. 248]{cover_thomas_2006}.

\begin{theorem}[Theorem 8.3.1 Cover and Thomas~\cite{cover_thomas_2006}]
	If the density $f(x)$ of a random variable $X$ is Riemann integrable, then
	\begin{align*}
	H(X^\Delta) + \log(\Delta) \rightarrow h(X), \text{ as }\Delta \rightarrow 0.
	\end{align*}
	Therefore the entropy of an n-bit quantisation of a continuous random variable $X$ is approximately $h(X) + n$, when using $\log_2$ in the expression above.
\end{theorem}

 Hence, if we quantise and apply Shannon entropy estimators we may be able to make a useful estimation of differential entropy, particularly with finer quantisations and more data. We extend this relationship to the case of Shannon entropy rate and differential entropy rate, which we clarify in the following theorem.

\begin{theorem}\label{shannon_differential_entropy_rate_link}
	If the joint density function, $f(x_1, ... , x_n)$, of the stochastic process, $\raisedchi = \{X_m\}_{m \in \mathbb{N}}$ is Riemann integrable $\forall n \in \mathbb{N}$, then
	\begin{align*}
	H(\raisedchi^\Delta) + \log(\Delta) \rightarrow h(f) = h(\raisedchi), \text{ as } \Delta \rightarrow 0.
	\end{align*}
	Therefore, the entropy rate of an n-bit quantisation of a continuous-valued, discrete-time stochastic process, $\raisedchi$ is approximately $h(\raisedchi) + n$, when using $\log_2$ in the expression above.
\end{theorem}

\begin{proof}
	For all finite $n \in \mathbb{N}$, the joint density of the finite collection of random variables is given by $f(x_1, ... , x_n)$. For notational convenience we use 
	\begin{align*}
		f(x_1, ... , x_n) &= f(\mathbf{x}),\\
		f(x_{i^1}, ... , x_{i^n}) &= f(\mathbf{x_i}),\\
		Pr(x_1, ... , x_n) &= Pr(\mathbf{x}),\\
		Pr(x_{i^1}, ... , x_{i^n}) &= Pr(\mathbf{x_i}),\\
		\text{and } dx_1 \ldots dx_n &= d\mathbf{x}.
	\end{align*}
	We partition the range of each random variable, $X$ of
        $\raisedchi$ into bins of length $\Delta$. Now, by the mean
        value theorem, for any probability density function $f(x)$
        there exists $x^{(i^j)} \in [i^j\Delta, (i^j + 1)\Delta)$ in each bin,
            such that 
	\begin{align*}
	f(x_{i^j})\Delta = \int_{i^j\Delta}^{(i^j + 1)\Delta} f(x) dx.
	\end{align*}
	Extending the result to the joint density $\exists \text{ } x_{i^1}, ... , x_{i^n}$, such that
	\begin{align*}
	f(x_{i^1}, ... , x_{i^n})\Delta^n = \int_{i^1\Delta}^{(i^1 + 1)\Delta} ... \int_{i^n\Delta}^{(i^n + 1)\Delta} f(\mathbf{x}) d\mathbf{x}.
	\end{align*}
	We define the quantised stochastic process, $\raisedchi^\Delta = \{X_m^\Delta\}_{m \in \mathbb{N}}$, by defining each of the quantised random variables $\forall i^j$ as,
	\begin{align*}
	X_{i^j}^\Delta = x_{i^j}, \text{ if } i^j\Delta \le X_{i^j} < (i^j + 1)\Delta.
	\end{align*}
	Then the probability that $\left( X_1^\Delta, ... , X_n^\Delta \right) = \left( x_{i^1}, ... , x_{i^n}\right)$, is given by,
	\begin{align*}
	Pr(\mathbf{x_i}) &= \int_{i^1\Delta}^{(i^1 + 1)\Delta} ... \int_{i^n\Delta}^{(i^n + 1)\Delta} f(\mathbf{x}) d\mathbf{x},\\
	&= f(\mathbf{x_i})\Delta^n.
	\end{align*}
	Hence, we consider the joint Shannon entropy of the quantised random variable, 
	\begin{align*}
	H(X_1^\Delta, ... , X_n^\Delta) &= -\sum_{i_1=-\infty}^{\infty}\ldots\sum_{i_n=-\infty}^{\infty} Pr(\mathbf{x_i}) \log(Pr(\mathbf{x_i})),\\
	&= -\sum_{i_1=-\infty}^{\infty}\ldots\sum_{i_n=-\infty}^{\infty} f(\mathbf{x_i})\Delta^n \log(f(\mathbf{x_i})\Delta^n),\\
	&= -\sum_{i_1=-\infty}^{\infty}\ldots\sum_{i_n=-\infty}^{\infty} f(\mathbf{x_i})\Delta^n \log(f(\mathbf{x_i}))\\
	&\qquad - \sum_{i_1=-\infty}^{\infty}\ldots\sum_{i_n=-\infty}^{\infty} f(\mathbf{x_i})\Delta^n \log(\Delta^n).
	\end{align*}
	As the joint density, $H(X_1^\Delta, ... , X_n^\Delta)$, is Riemann integrable $\forall n \in \mathbb{N}$, this implies that as $\Delta \rightarrow 0$,
	\begin{align*}
		-\sum_{i_1=-\infty}^{\infty}\ldots\sum_{i_n=-\infty}^{\infty} &f(\mathbf{x_i})\Delta^n \log(f(\mathbf{x_i}))\\
		 &= -\int_{-\infty}^{\infty}\ldots\int_{-\infty}^{\infty} f(\mathbf{x}) \log f(\mathbf{x}) d\mathbf{x},\\
		&= h(X_1, \ldots , X_n).
	\end{align*}
	Since $f(\mathbf{x_i})$ is a joint density over $\mathbb{R}^n$, this implies that,
	\begin{align*}
		\sum_{i_1=-\infty}^{\infty}\ldots\sum_{i_n=-\infty}^{\infty} f(\mathbf{x_i})\Delta^n &= \int_{-\infty}^{\infty}\ldots\int_{-\infty}^{\infty} f(\mathbf{x_i}) d\mathbf{x}\\
		&= 1.
	\end{align*}
	Therefore, we get as $\Delta \rightarrow 0$,
	\begin{align*}
		H(X_1^\Delta, \ldots , X_n^\Delta) + n\log(\Delta) \rightarrow h(X_1, \ldots , X_n).
	\end{align*}
	Now we consider the limit of the quantised joint differential entropy, as the length of the joint entropy tends to infinity,
	\begin{align*}
	\lim\limits_{n \rightarrow \infty} \frac{1}{n} \left(H(X_1^\Delta, \ldots , X_n^\Delta) + n\log(\Delta) \right)  &\rightarrow\\ \lim\limits_{n \rightarrow \infty} \frac{1}{n} &\left(h(X_1, \ldots , X_n) \right),
	\end{align*}
	and substituting the definitions of Shannon and differential entropy rate respectively, we get
	\begin{align*}
	\implies H(\raisedchi^\Delta) + \log(\Delta) &\rightarrow h(\raisedchi).
	\end{align*}
\end{proof}

There will, however, be an error which is due to the
difference between the real differential entropy rate, and the
approximation of the integral at the quantisation size,
$\Delta$.

Another thing to note is that as the quantisation window gets finer,
\emph{i.e.}, as $\Delta$ gets smaller, the term in the estimator, $\log(\Delta)
\rightarrow -\infty$. There is a potential concern that this disparity
between the correction term and the actual estimate could lead to
numerical errors. Thus, smaller quantisations aren't necessarily
better for practical estimation. 


%% file: quantised_estimation_TSP.tex
\section{NPD-Entropy Estimator}\label{quantised_estimation}

Theorem~\ref{shannon_differential_entropy_rate_link} gives us a way to
convert an estimator of entropy rate for discrete-valued,
discrete-time stochastic processes into an estimator for a continuous-valued process. We call this strategy NPD-Entropy estimation and
explore its properties in this section. In particular we develop links
from the properties of the Shannon entropy rate estimator that is
selected to the corresponding NPD-estimator.

\begin{definition}[NPD-Entropy]
	The NPD-Entropy estimator of the differential entropy rate,
        $h(\raisedchi)$, of a continuous-valued, discrete time
        stochastic process, $\raisedchi = \{X_i\}_{i \in \mathbb{N}}$,
        using a Shannon entropy rate estimator,
        $H(\raisedchi^\Delta)$, of the corresponding quantised
        discrete-time, discrete-valued stochastic process
        $\raisedchi^\Delta=\{X_i^\Delta\}_{i \in \mathbb{N}}$ with
        window size, $\Delta$, is defined as
	\begin{align*}
		\hat{h}_{NPD}(\raisedchi) = \hat{H}(\raisedchi^\Delta) + \log\Delta.
	\end{align*}
\end{definition}

We analyse some properties of NPD-Entropy and start by considering the consistency of the estimation technique.

\begin{theorem}\label{consistency}
	NPD-Entropy, $\hat{h}_{NPD}$, is a consistent estimator as
        $\Delta \rightarrow 0$, if and only if the associated Shannon
        entropy rate estimator, $\hat{H}$, is consistent, \ie as
        $\Delta \rightarrow 0$ the following two are equivalent:
	\begin{align*}
	\lim\limits_{n \rightarrow \infty}
        \hat{H}(\raisedchi^\Delta_n) & \rightarrow H(\raisedchi^\Delta) \, a.s. , \\ 
	 \lim\limits_{n \rightarrow \infty} \hat{h}_{NPD}(\raisedchi_n) & \rightarrow h(\raisedchi) \, a.s.,
        \end{align*}
        where $n$ denotes the length of the data sequence to which the
        estimator is applied.
\end{theorem}

\begin{proof}
  By definition
  \begin{eqnarray*}
      \lim\limits_{n \rightarrow \infty} \hat{h}_{NPD}(\raisedchi_n)
      & = &  \lim\limits_{n \rightarrow \infty} \hat{H}(\raisedchi^\Delta_n) + \log\Delta \\
      & = &  \log\Delta + H(\raisedchi^\Delta).
  \end{eqnarray*}
  Hence, by Theorem~\ref{shannon_differential_entropy_rate_link}, the
  limiting properties of the two are linked. 
\end{proof}


\begin{remark}
	This is a very general result, that only relies on the existence of a consistent estimator for Shannon entropy rate. The argument is agnostic to the mode of convergence used, and the strength of the convergence depends on which mode of convergence is used for the consistency of the Shannon entropy rate estimator. 
\end{remark}

 We explore more general results for the variance, mean squared error, and then apply these to efficiency, as efficiency is applied to unbiased estimators of the entropy rate. 

\begin{theorem}\label{bias_var_mse}
	The NPD-Entropy, $\hat{h}_{NPD}$, constructed from an estimator
        of the Shannon entropy rate, $\hat{H}$ has the following properties
	\begin{enumerate}
		\item $\mbox{Bias}_h[\hat{h}_{NPD}(\raisedchi_n)] = \mbox{Bias}_H[\hat{H}(\raisedchi^\Delta_n)]$,
		\item $\mbox{Var}(\hat{h}_{NPD}(\raisedchi_n)) = \mbox{Var}(\hat{H}(\raisedchi^\Delta_n))$,
		\item $\mbox{MSE}(\hat{h}_{NPD}(\raisedchi_n)) = \mbox{MSE}(\hat{H}(\raisedchi^\Delta_n))$,
	\end{enumerate}
        as $\Delta \rightarrow 0$.
\end{theorem}

\begin{proof}
	From the definition of bias~\cite[pg. 126]{rice2006mathematical}, we have
	\begin{align*}
	\mbox{Bias}_H[\hat{H}] &= E[\hat{H}] - H,\\
	&= E[(\hat{h}_{NPD} - \log\Delta)] - (h - \log\Delta), \text{ as } \Delta \rightarrow 0\\
	&= E[\hat{h}_{NPD}] - h,\\
	&= \mbox{Bias}_h[\hat{h}_{NPD}].
	\end{align*}
	The result for the mean squared error of the estimator follows by an identical argument, substituting, $\hat{h}_{NPD} - \log\Delta$ and $h - \log\Delta$ for $\hat{H}$ and $H$ respectively as $\Delta \rightarrow 0$, therefore we omit the argument.\\
	We consider another characterisation of the mean squared error, known as the bias-variance decomposition~\cite[pg. 24]{friedman2001elements}, 
	\begin{align*}
	\mbox{MSE}[\hat{H}] &= \mbox{Bias}_H[\hat{H}]^2 + \mbox{Var}(\hat{H}),
	\end{align*}
	and placing in terms of the variance and substituting the bias and mean square error equivalences as $\Delta \rightarrow 0$, we get
	\begin{align*}
	\mbox{Var}(\hat{H}) &= \mbox{Bias}_H[\hat{H}]^2 - \mbox{MSE}[\hat{H}],\\
	&= \mbox{Bias}_h[\hat{h}_{NPD}]^2 - \mbox{MSE}[\hat{h}_{NPD}], \text{ as }  \Delta \rightarrow 0,\\
	&= \mbox{Var}(\hat{h}_{NPD}).
	\end{align*}
\end{proof}

We investigate the asymptotic normality of the Shannon entropy rate estimator, That is, the asymptotic distribution of the estimator being a normally distributed random variable. 

\begin{theorem}
	NPD-Entropy, $h$, is asymptotically normal as $\Delta \rightarrow 0$ if the associated Shannon entropy rate estimator, $H$, is consistent and asymptotically normal.
\end{theorem}

\begin{proof}
	Given a $\sigma^2 \in \mathbb{R^+}$, for $n \rightarrow \infty$ with an asymptotically normal estimator we have,
	\begin{align*}
	\sqrt{n} \left(\hat{H}(\raisedchi^\Delta_n) - H \right) &\rightarrow \mathcal{N}(0, \sigma^2).
	\end{align*}
	Substituting the expression for the differential entropy rate,we get 
	\begin{align*}
	\sqrt{n} \left( (\hat{h}_{NPD}(\raisedchi_n) - \log \Delta) - (h - \log \Delta) \right) &\rightarrow\\ \mathcal{N}(0, \sigma^2), &\text{ as }\Delta \rightarrow 0,
	\end{align*}
	and we conclude
	\begin{align*}
	\sqrt{n} \left(\hat{h}_{NPD}(\raisedchi_n)- h \right) &\rightarrow \mathcal{N}(0, \sigma^2).
	\end{align*}
\end{proof}

The next property we examine is the efficiency, which is,
\begin{align*}
	e(T) = \frac{I(\theta)^{-1}}{Var(T)},
\end{align*}
for an unbiased estimator, $T$, of a parameter, $\theta$, where $ I(\theta) = E\left[ \left(\frac{\partial}{\partial \theta} \log f(T;\theta)\right)^2 \mid \theta\right]$
 is the Fisher information~\cite[pg. 481]{cramer1999mathematical}. Efficiency of an estimator measures the variance of the estimator with respect to the lowest possible variance, the Cramer-Rao bound~\cite[pg. 480]{cramer1999mathematical}, which is when $e(T) = 1$.

\begin{theorem}
	The efficiency of NPD-Entropy, $e(\hat{h}_{NPD}(\raisedchi_n))$ is equal to the efficiency of the Shannon entropy rate estimator, $e(\hat{H}(\raisedchi^\Delta_n))$ for a Riemann integrable probability density function for the process, $\raisedchi = \{X_i\}_{i \in \mathbb{N}}$.
\end{theorem}

\begin{proof}
	We consider the Fisher information of the quantised process,
	\begin{align*}
	I_{\raisedchi^\Delta}(H) &= \sum_{i} \left( \frac{d}{d H} \log p_i \right)^2 p_i,\\
	&= \sum_{i} \left( \frac{d}{d H} \log (f(x_i) \Delta) \right)^2 f(x_i) \Delta.
	\end{align*}
	Where $p_i = \int_{i^j\Delta}^{(i^j + 1)\Delta} f(x) dx = f(x_{i^j})\Delta$, since the density is Riemann integrable. Therefore the Fisher information is
	\begin{align*}
	I_{\raisedchi^\Delta}(H) &= \sum_{i} \left( \frac{d}{d H} \log f(x_i) + \log \Delta \right)^2 f(x_i) \Delta,\\
	&= \sum_{i} \left( \frac{d}{d H} \log f(x_i) \right)^2 f(x_i) \Delta,
	\end{align*}
	since $\log \Delta$ is constant with respect to $H$. Then 
	\begin{align*}
	I_{\raisedchi^\Delta}(H) &= \int \left( \frac{d}{d h} \log f(x) \right)^2 f(x) dx, \text{ as } \Delta \rightarrow 0,\\
	&= I_{\raisedchi}(h).
	\end{align*}
	Where $\frac{dH}{df} = \frac{d(h - \log\Delta)}{df} = \frac{dh}{df}$, for a given function $f$, as $\Delta \rightarrow 0$.\\
	Now we, consider the efficiency of the differential entropy rate estimator,
	\begin{align*}
	e(\hat{h}_{NPD}) &= \frac{\frac{1}{I(h)}}{\mbox{Var}(\hat{h}_{NPD})},\\
	&= \frac{\frac{1}{I(H)}}{Var(\hat{H})}, \text{ from Theorem~\ref{bias_var_mse}},\\
	&= e(\hat{H})
	\end{align*}
\end{proof}

Another property of interest is the convergence rate of an estimator, as we often want to understand how long it would take to be within a certain distance from the real value, to understand how much data is required for estimation. 

\begin{theorem}\label{theorem_convergence_rate}
	The convergence rate of NPD-Entropy, $\hat{h}_{NPD}(\raisedchi_n)$  is equal to the convergence rate of the Shannon entropy rate estimator, $\hat{H}(\raisedchi^\Delta_n)$ as $\Delta \rightarrow 0$.
\end{theorem}

\begin{proof}
	For the convergence rate, $C(n)$, of a Shannon entropy rate estimator, we have $\forall n \in \mathbb{N}$,
	\begin{align*}
	Pr(|\hat{H}(\raisedchi^\Delta_n) - H| > \epsilon) &\le C(n).
	\end{align*}
	Then making the substitution of the differential entropy rate estimator, and cancelling out the $\log \Delta$ terms, we get
	\begin{align*}
	Pr(|\hat{h}_{NPD}(\raisedchi_n) - h| > \epsilon) &\le C(n).
	\end{align*}
\end{proof}

\subsection{Implementation and Initial Results}

\begin{figure*}[t]
	\centering
        \begin{subfigure}[t]{0.98\columnwidth}
        \centering
	  \includegraphics[width=0.8\linewidth]{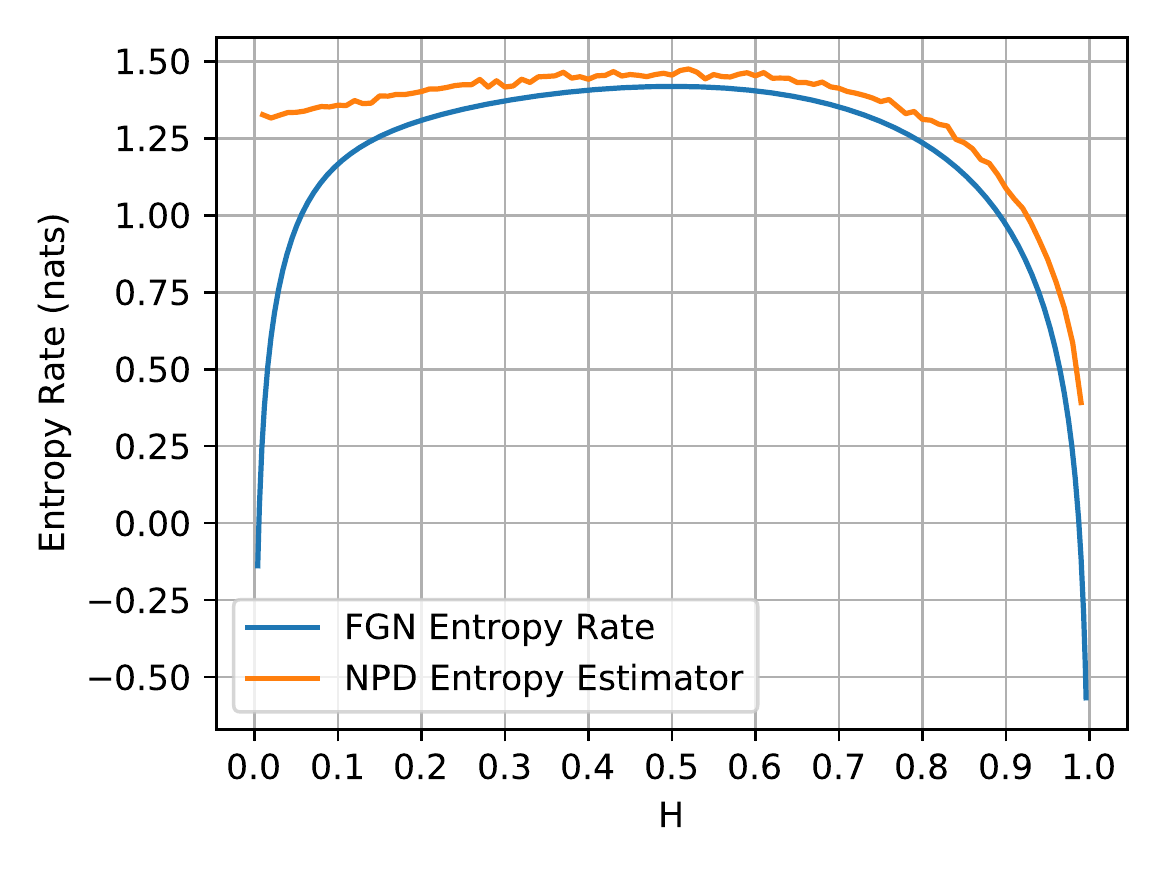}
	  \caption{Comparison for Fractional Gaussian Noise. The estimates
            have a small bias, but get poorer as $H \rightarrow 0$ and
            the entropy rate function asymptotically tends to
            $-\infty$. }
	  \label{fig: fgn_entropy_rate_comparison}
        \end{subfigure}
        \hfill
        \begin{subfigure}[t]{0.98\columnwidth}
        \centering
 	\includegraphics[width=0.8\linewidth]{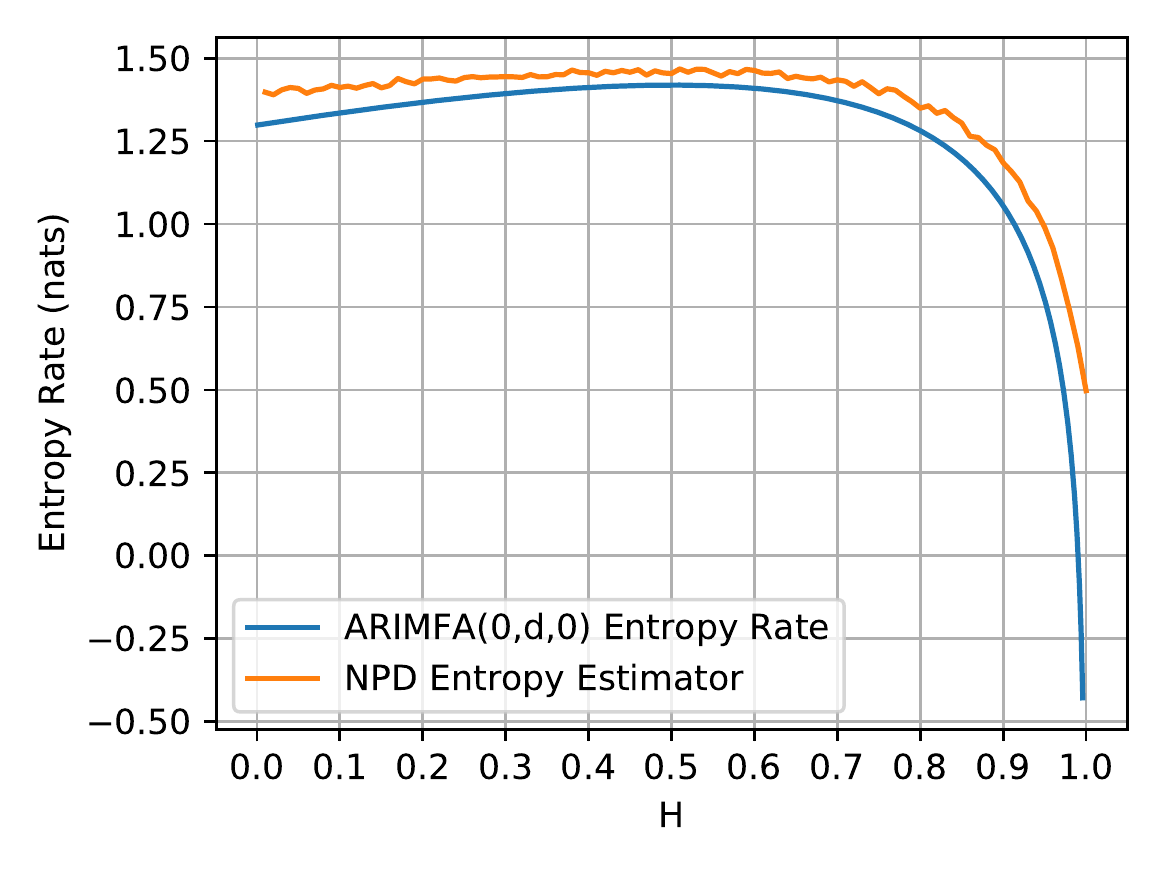}
	\caption{Comparison for ARFIMA(0,d,0). Note the close
          correspondence, with a small bias.}
	\label{fig: arfima_entropy_rate_comparison}
        \end{subfigure}
        \caption{NPD Entropy with $\Delta = 1$ compared to actual entropy rate of FGN and ARFIMA(0,d,0). With $\Delta = 1$, the Shannon entropy rate estimator is directly used as the NPD-Entropy estimate.}
\end{figure*}

To test this concept, we have implemented an estimator using the discrete-valued estimator described in Section~\ref{discrete_entropy}. The conditions in Kontoyiannis and
Suhov~\cite{kontoyiannis1998nonparametric} were used to ensure that we produce a
consistent Shannon entropy rate estimator.

The estimator was implemented in Python using the NumPy
library~\cite{harris2020array}.  We parallelised the
implementation of NPD Entropy to increase the performance of the
estimation technique since we can calculate the length of each
prefix sequence, $L^n(x)$, independently. This has been implemented using the Python
library Numba~\cite{numba2015}, which is able to parallelise the loop
to calculating the prefixes, $L(x)$. The algorithm can be run online, for quick estimation of entropy rates of streaming data. 

The estimators were tested on data generated by FGN and ARFIMA(0,d,0)
processes. We generated 50 test samples of 2000 data points for each
process. The estimates for each process were then averaged, to show
the mean estimate at each value of $H$.

From observing both Figure~\ref{fig: fgn_entropy_rate_comparison} and
Figure~\ref{fig: arfima_entropy_rate_comparison}, utilising a
quantised interval size of 1 (i.e. $\Delta = 1$), we see overall good agreement,
except in the region $H \rightarrow 0$ for FGN. In all cases the
NPD-Entropy estimator was better at picking up the underlying trends
of the entropy rate functions than
the relative measures tested in Section~\ref{testing_measures}, but
specific entropy produced closer estimates of the differential entropy
rate at a large cost in computation time. 

The limit theorems and the results on the inherited properties, from Section~\ref{quantised_estimation}, of the estimator hold as $\Delta \rightarrow 0$, which for practical purposes we are unable to achieve. However, this leaves us with trying to understand the possible errors that have been introduced. Potential sources of estimation error are from the difference of the Riemann sums and the value of the integral, \emph{i.e.}, as $\sum f(x) \rightarrow \int f(x)$. As the quantisation window decreases, as $\Delta \rightarrow 0$, small differences from the true value and estimated values can be enlarged by the term $\log\Delta$. To test the deviation from true value, we produced estimates for a range of different quantisation window sizes, $\frac{1}{3}, \frac{1}{2}, 1, 2, 3$, for the Hurst parameter range, $[0,1]$. Given that we have standardised all of the process variances to be $\sigma^2 = 1$, the window size is in direct proportion with the variance. This means that small fluctuations of the process will stay within same bin, and therefore the measured uncertainty will come from larger shifts. In general we would tailor the window size to the variance, since we want to balance the number of bins and the expected size of movements between random variables of the process.

The results for FGN are shown in Figure~\ref{fig: fgn_entropy_rate_comparison_diff_deltas}. The most accurate of these across the whole range is at $\Delta = 1$, which is surprising as we should be more accurate as the quantisation size decreases. However, we are balancing a few different sources of error and in this case the best result is to quantise at a window size of 1. However, as the Hurst parameter tends to 1, the finer quantisations become more accurate estimators. As the correlations become stronger the variation in the process will be more subtle, hence the smaller quantisation windows are required to estimate the uncertainty.
      
The results are similar for ARFIMA(0,d,0) processes, with the closest estimator, across the entire parameter range, being when the quantisation window was of size, $\Delta = 1$. The estimators got worse as $\Delta$ got further from 1, i.e. the next closest were $\Delta = \frac{1}{2}$ and $2$. The finer quantisations, $\Delta = \frac{1}{2}, \frac{1}{3}$ became more accurate as $H$ increased.

\begin{figure*}[t]
  \centering
	\begin{subfigure}[t]{0.98\columnwidth}
		\centering
		\includegraphics[width=0.85\linewidth]{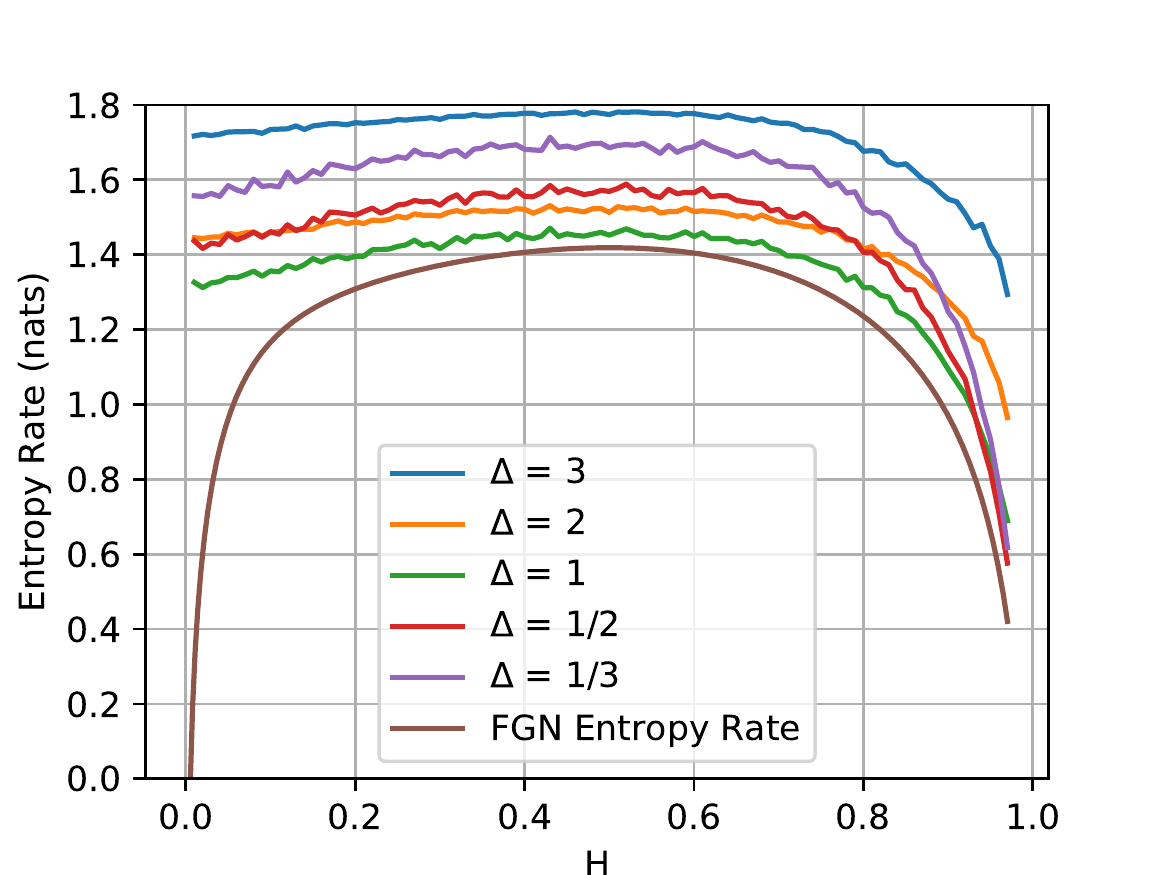}
		\caption{Comparison of the quantised estimates and true values of differential entropy rate of FGN. The window size of one has the best estimates, however as the Hurst parameter tends to one the finer quantisations, \emph{i.e.}, $\frac{1}{3}$ and $\frac{1}{2}$, approximate the asymptotic decrease better.}
		\label{fig: fgn_entropy_rate_comparison_diff_deltas}
	\end{subfigure}
	\hfill
	\begin{subfigure}[t]{0.98\columnwidth}
		\centering
		\includegraphics[width=0.85\linewidth]{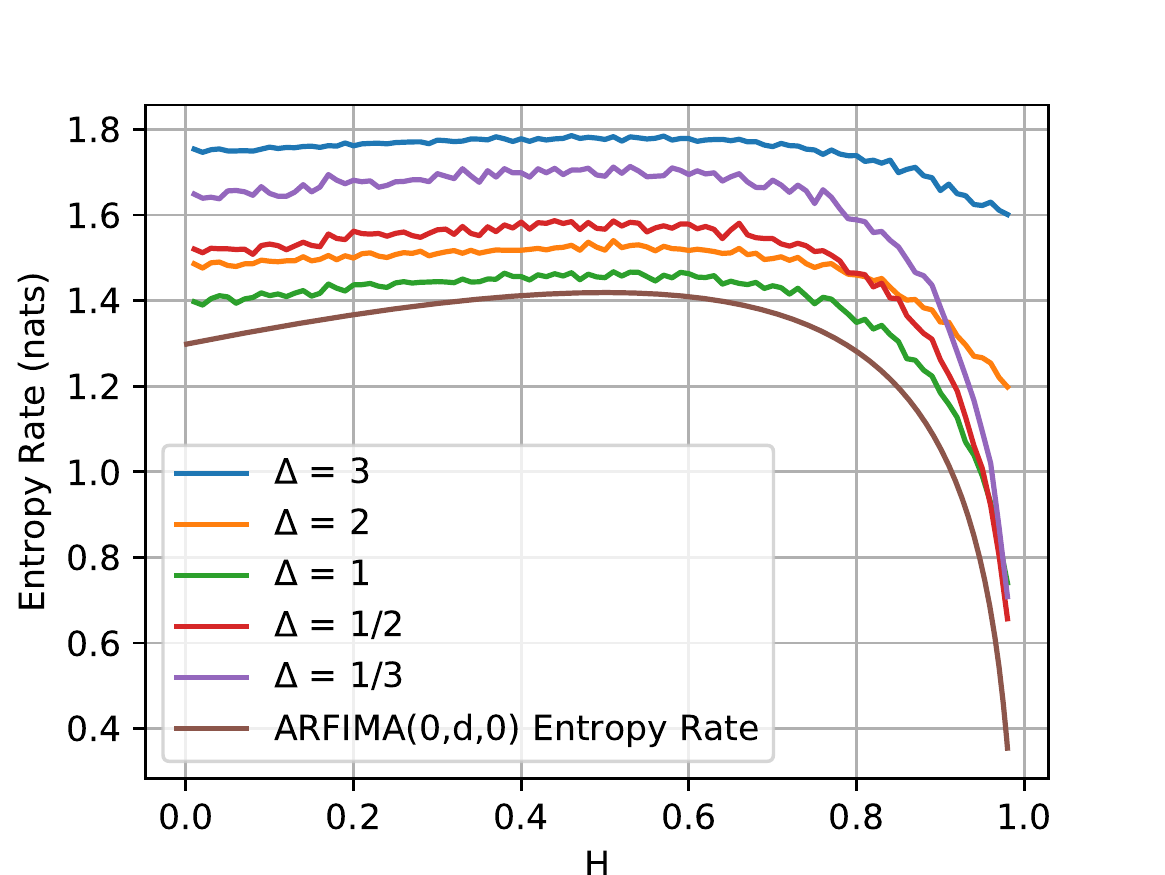}
		\caption{Comparison of the quantised estimates and true values of differential entropy rate of ARFIMA(0,d,0). Similar to FGN, the window size of $1$ provides the best estimates over the range of Hurst parameter values, however we also see that the finer quantisations become more accurate as $H$ increases towards $1$.}
		\label{fig: arfima_entropy_rate_comparison_diff_deltas}
	\end{subfigure}
	\caption{NPD Entropy with $\Delta = \frac{1}{3}, \frac{1}{2}, 1, 2, 3$ compared to actual entropy rate of FGN and ARFIMA(0,d,0)}
\end{figure*}


\subsection{Robustness of estimator to non-stationarity}

A common assumption that is made to enable analysis of stochastic processes and for data generating processes, is stationarity, meaning that the probability distribution is time invariant over the stochastic processes.

\begin{definition}
	A stochastic process, $\raisedchi = \{X_i\}_{i \in \mathbb{N}}$, with a cumulative distribution function for its joint distribution at times, $t_1 + \tau, \ldots, t_n + \tau$, of, $F_X(t_1 + \tau, \ldots, t_n + \tau)$. We say that the process is stationary if $F_X(t_1 + \tau, \ldots, t_n + \tau) = F_X(t_1, \ldots, t_n), \forall \tau > 0, t_1, \ldots, t_n \in \mathbb{R}, \forall n \in \mathbb{N}$.
\end{definition}

Many real world processes are not really stationary, hence we want to analyse how robust the quantised estimation technique is to non-stationarity. In the previous section we have shown that the estimator has good performance in the cases of strongly correlated processes. However, we would like to test the performance against processes that have a convergent entropy rate, but have moment statistics that vary with time. 


The initial test for a departure from stationarity is a
mean-shifting process. We define this as $X_n = \mu_n + \epsilon_n$,
where $\epsilon_n \sim \mathcal{N}(0,1)$ and independent and where $\mu_n$ alternates periodically between 0 and 1
every 100 samples.The mean shifts are predictable so the entropy rate
is
\begin{align*}
	 h\left(\mu_n + \epsilon_n\right) 
	 &= h\left(\epsilon_n\right)  = \frac{1}{2} \log(2\pi  e) \approx 1.419.
\end{align*}

We test this process with 50 realisations of 2000 data points. The mean and variance of
the estimates are shown in Table~\ref{tab:non_stationary}. NPD-Entropy, with a choice of parameter $\Delta = 1$, made an estimate of the entropy rate of 1.554,
which overestimates the true value by a small amount, similar to
specific entropy. Sample and permutation entropy both provided large
overestimates of the true value, with permutation entropy having an
extremely low variance, close to its maximum value, $\log(3!) \approx
2.585$. Approximate entropy was a large underestimate, similar to the
behaviour with $m=3$ for both the ARFIMA(0,d,0) and FGN
estimates.

The second test for robustness to a departure from non-stationarity is
to consider a simple model that has a stationary mean, but varying
second order statistics, a Gaussian Random Walk, $\{Z_n\}_{n \in
  \mathbb{N}}$. Each step, $X_i \sim \mathcal{N}(0, 1)$ is independent and we
consider the sum of the steps $Z_n = \sum_{i = 1}^{n} X_i$. The
entropy rate can be derived by considering the limiting conditional
entropy to give $h(Z_n) = h(X_n) \approx 1.419$. However,
the process has non-stationary second moment $\mbox{Var}(Z_n) = n$.


These initial results suggest that specific entropy is relatively
robust, but its alternatives are not. However NPD-Entropy is somewhat
robust to mean shifts, if not to changes in variance. The variance in this case is constantly growing in proportional to the number of random samples of the process, as opposed to the previous process which had non-stationarity in its mean, but a constant variance of the normally distributed noise. This may explain the difficulty of capturing the true value in this case, as all of the bins are likely to have small numbers of observations, unless large amounts of data points are observed. Therefore, we have to be careful in how processes depart from stationarity, to understand how closely we can estimate the true entropy rate.

\begin{table}
	\begin{center}
		\begin{tabular}{l|d{1.3}d{1.6}|d{1.3}d{3.4}}
			 & \multicolumn{2}{c|}{\textbf{Mean Shift}} & \multicolumn{2}{c}{\textbf{Gaussian Walk}} \\
			 \textbf{Estimation Technique} & \multicolumn{1}{r}{ \textbf{Mean}} & \multicolumn{1}{r|}{ \textbf{Variance}} & \multicolumn{1}{r}{ \textbf{Mean}} & \multicolumn{1}{c}{\textbf{Variance}} \\
			\hline
			Approximate Entropy & 0.482 & 0.0008 & 0.099 & 0.0007 \\
			Sample Entropy & 2.263 & 0.015 & 2.336 & 0.205\\
			Permutation Entropy & 2.582 & 0.000005 & 2.497 & 0.0002 \\
			Specific Entropy & 1.478 & 0.0005 & 1.523 & 0.0007 \\
			NPD-Entropy & 1.554 & 0.002 & 2.385 & 0.083\\
		\end{tabular}
		\caption{Differential entropy rate estimates applied
                  to two non-stationary processes: a Gaussian
                  i.i.d. mean shift and a Gaussian walk with entropy
                  rates 1.419. Each process was simulated 50
                  times. The relative measures, the sample entropy and
                  permutation entropy both provide large overestimates
                  of the true entropy, and approximate entropy
                  provides large underestimates. NPD-entropy behaves
                  only slightly worse than specific entropy on the
                  mean-shift process, but somewhat worse on the
                  Gaussian walk. Estimator variances are included to
                  inform about the relative size of the errors. }
		\label{tab:non_stationary}
	\end{center}
\end{table}

\subsection{Complexity Analysis of Estimation Techniques}



The worst-case asymptotic time complexities of the estimators tested
here are shown in Table~\ref{tab:complexity}. Approximate entropy has
two loops which run to calculate the quantity, first to calculate
every $C_i^m(r)$, the number of strings that exceed a threshold, by
fixing one of the substrings, then a second loop over all of the
$C_i^m(r)$'s, \emph{i.e.}, for a pairwise comparison of all
substrings. Hence, the time complexity is approximately $O(N^2)$,
where $N$ is the length of data. Sample entropy requires two separate
loops, which run over the substrings of order $m$ and $m+1$ all
calculating the number of pairs that exceed a threshold, \ie two loops of order $N^2$, which also
results in a pairwise comparison and hence the time complexity is
approximately $O(N^2)$. Permutation entropy takes the relative
frequency of each possible permutation, of order $n$, of which there
exist $n!$ permutations. This is performed across the entire data set,
and hence the worst-case time complexity for permutation entropy is
$O(n!N)$, and hence is extremely sensitive to the choice of order
$n$. With appropriate choices of order this can be used as an quick
measure of complexity of time series, however it does grow extremely
quickly with the order of permutation and hence there is a trade off
here between how much data is available with what order can be
selected for this technique.

Specific entropy's complexity depends on the choice of kernel density
estimator. We used their implementation in R which used the package
{\tt np} for non-parametric kernel density
estimation~\cite{darmon2016specific}. The technique calculates the
bandwidth parameters and then calculates the product and univariate
kernel estimates for windows of length $p$. Hence, to make an estimate
we do $N$ kernel density evaluations, and therefore $Np$ for the
calculation of the estimate of the joint density per window, and hence
with a loop over all of the windows we get an approximate complexity
of $O(N^2p)$. Although this does not appear substantially worse that
the other estimators, the actual amount of computation is dramatically
larger.

NPD Entropy has two steps: (i) quantisation and (ii) a
Shannon entropy rate estimator, which in most cases will dominate the
performance. In our case we use the Grassberger estimate, which
utilises string-matching based on Lempel-Ziv
algorithm~\cite{ziv1977universal}, which has a complexity of $O(N\log
N)$, using the algorithm of
Grassberger~\cite{grassberger1989estimating}.

Complexity estimates are useful to understand scaling properties, but
the actual computation times for finite data can be dominated by
non-asymptotic terms, and so we test performance directly using 1000
estimates of $N=1000$ data points. Tests were performed on a 2.4GHz
Intel Core i5 processor with 8GB of RAM, running MacOS 10.14.6. The
results are shown in Table~\ref{tab:run_time_comparison}. Approximate
entropy was much slower than sample entropy, with permutation entropy
being the fastest by far. However NPD-Entropy is faster than all but
permutation entropy and runs extremely quickly compared to specific
entropy.  The specific entropy was considerably slower than all the
measures with the 1000 estimates being in the order of 6 days.




Given the time complexity, we suggest using NPD-Entropy as an
efficient way of classifying complexity of a system. NPD-Entropy
provides a more accurate and robust measurement with comparable or
better computation times than most alternatives. Specific entropy is
still superior for accurate measures, but computationally prohibitive
except where accuracy is the only criteria.

%% file: conclusion.tex
\section{Conclusion}
We have defined a new technique for the non-parametric estimation of the differential entropy rate. We have made an explicit link between the differential entropy rate and the Shannon entropy rate of the associated quantised process. This forms the basis of the estimation technique, NPD Entropy, by quantising the continuous data and utilising the existing theory of Shannon entropy rate estimation. We have shown that this estimation technique inherits statistical properties from the Shannon entropy rate estimator and performs better than other differential entropy rate estimators in the presence of strongly correlated data. 

In addition, we have investigated the robustness to non-stationarity of the estimation technique, and provided over-estimates of the entropy. More research is required to analyse the robustness to non-stationarity and understand if the influence is similar to the Shannon entropy rate case. Finally, we have demonstrated the utility of NPD Entropy and shown that it can be run quickly and efficiently to make decent estimates if used in an online mode.